\documentclass[prl, twocolumn, superscriptaddress, hyperref]{revtex4}
\pdfoutput=1
\usepackage{enumitem}
\usepackage{graphicx}
\usepackage{bm}
\usepackage{amsmath}
\usepackage{amssymb}
\usepackage{underscore}
\usepackage{color}
\usepackage{float}
\usepackage[caption = false]{subfig}
\usepackage{amsmath,amsthm,amsfonts,amssymb,amscd}

\usepackage{hyperref}
\hypersetup{
    colorlinks=true,
    linkcolor=blue,
    filecolor=magenta,      
    urlcolor=violet,
}

\usepackage{dsfont}
\def\squareforqed{\hbox{\rlap{$\sqcap$}$\sqcup$}}
\def\qed{\ifmmode\squareforqed\else{\unskip\nobreak\hfil
\penalty50\hskip1em\null\nobreak\hfil\squareforqed
\parfillskip=0pt\finalhyphendemerits=0\endgraf}\fi}

\def\duzomniejsze{<\kern-.7mm<}
\def\duzowieksze{>\kern-.7mm>}

\def\textbf#1{{\bf #1}}
\def\beq{\begin{equation}}
\def\eeq{\end{equation}}
\def\be{\begin{equation}}
\def\ee{\end{equation}}
\def\ben{\begin{eqnarray}}
\def\een{\end{eqnarray}}
\def\beqa{\begin{eqnarray}}
\def\eeqa{\end{eqnarray}}
\def\eea{\end{array}}
\def\bea{\begin{array}}

\newcommand{\bei}{\begin{itemize}}
\newcommand{\eei}{\end{itemize}}
\newcommand{\bee}{\begin{enumerate}}
\newcommand{\eee}{\end{enumerate}}
\newcommand{\nc}{\newcommand}
\nc{\1}{{\openone}}

\def\>{\rangle}
\def\<{\langle}

\newtheorem{lemma}{Lemma}

\newtheorem{prop}{Proposition}
\newtheorem{theorem}{Theorem}
\newtheorem{dfn}{Definition}

\newtheorem*{definition}{Definition}

\def\bed{\begin{definition}}
\def\eed{\end{definition}}
\def\bel{\begin{lemma}}
\def\eel{\end{lemma}}
\def\bet{\begin{theorem}}
\def\eet{\end{theorem}}
\def\be{\begin{equation}}
\def\ee{\end{equation}}

\allowdisplaybreaks

\usepackage{amsmath,amssymb,amsthm}
\usepackage{enumerate}
\usepackage{tikz}
\usepackage{blkarray}

\usepackage{mathtools}

\begin{document}

\title{Large violations in Kochen Specker contextuality and their applications}
\author{Ravishankar Ramanathan}
\affiliation{Department of Computer Science, The University of Hong Kong, Pokfulam Road, Hong Kong}

\author{Yuan Liu}
\affiliation{Department of Computer Science, The University of Hong Kong, Pokfulam Road, Hong Kong}

\author{Pawe{\l} Horodecki}
\affiliation{International Centre for Theory of Quantum Technologies, University of Gda\'{n}sk, Wita Stwosza 63, 80-308 Gda\'{n}sk, Poland}
\affiliation{Faculty of Applied Physics and Mathematics, National Quantum Information Centre, Gda\'{n}sk University of Technology, Gabriela Narutowicza 11/12, 80-233 Gda\'{n}sk, Poland} 

\begin{abstract}
It is of interest to study how contextual quantum mechanics is, in terms of the violation of Kochen Specker state-independent and state-dependent non-contextuality inequalities.
We present
state-independent non-contextuality inequalities with large violations, in
particular, we exploit a connection between Kochen-Specker proofs and
pseudo-telepathy games to show KS proofs in Hilbert spaces of dimension $d \geq
2^{17}$ with the ratio of quantum value to classical bias being
$O(\sqrt{d}/\log d)$. We study the properties of this KS set and show applications of the large violation. It has been
recently shown that Kochen-Specker proofs always consist of substructures of
state-dependent contextuality proofs called $01$-gadgets or bugs. We show a
one-to-one connection between $01$-gadgets in $\mathbb{C}^d$ and Hardy
paradoxes for the maximally entangled state in $\mathbb{C}^d \otimes
\mathbb{C}^d$. We use this connection to construct large violation $01$-gadgets
between arbitrary vectors in $\mathbb{C}^d$, as well as novel Hardy paradoxes
for the maximally entangled state in $\mathbb{C}^d \otimes \mathbb{C}^d$, and
give applications of these constructions. As a technical result, we show that
the minimum dimension of the faithful orthogonal representation of a graph in
$\mathbb{R}^d$ is not a graph monotone, a result that that may be of
independent interest.
	
\end{abstract}

\maketitle


\section{Introduction}
The Kochen-Specker (KS) theorem \cite{KS} is a central result in quantum foundations, that states that in any Hilbert space $\mathbb{C}^d$ of dimension $d \geq 3$, there exist finite sets of projectors that do not admit of a deterministic non-contextual value assignment. By non-contextual is meant the property that the value assignments to the projectors do not depend on the context, i.e., the other compatible projectors with which they are measured. Quantum contextual correlations have been shown to be a useful resource in several information processing tasks \cite{Howard, HHHH+10, CLMW10}.

The KS theorem is usually proved by showing a finite set of vectors $\mathcal{S} = \{|v_1 \rangle, \dots, |v_n \rangle\} \subset \mathbb{C}^d$ for which no deterministic assignment $f: S \rightarrow \{0,1\}$ is possible satisfying the KS conditions that $(i) \sum_{|v \rangle \in \mathcal{O}} f(|v \rangle) \leq 1$ for every set $\mathcal{O} \subset \mathcal{S}$ of mutually orthogonal vectors, and $(ii) \sum_{|v\rangle \in \mathcal{B}} f(|v \rangle) = 1$ for every complete basis $\mathcal{B} \subset \mathcal{S}$ of $d$ mutually orthogonal vectors. An assignment $f$ satisfying the two conditions is called a $\{0,1\}$-coloring of $\mathcal{S}$ and KS vectors sets are said to be $\{0,1\}$-uncolorable. A generalization of the KS sets due to Renner and Wolf \cite{RW04} called \textit{weak Kochen Specker sets} also serves to prove the KS theorem. A weak KS set \cite{RW04} is a set of (unit) vectors $S \subset \mathbb{C}^d$ such that for any function $f: S \rightarrow \{0,1\}$ satisfying $\sum_{|u \rangle \in b} f(|v \rangle) = 1$ for all orthogonal bases $b \subset S$, there exist two orthogonal unit vectors $|u_1 \rangle \perp |u_2 \rangle \in S$ such that $f(|u_1 \rangle) = f( |u_2 \rangle) = 1$. The quantum contextual correlations from the KS proofs lead to violation of non-contextuality inequalities by any quantum state in $\mathbb{C}^d$ leading to the notion of state-independent contextuality. Smaller sets of projectors can also be found to demonstrate contextuality. For these, the violation of the corresponding non-contextuality inequality occurs only when the projectors are measured on particular states in $\mathbb{C}^d$, we say that such sets give rise to state-dependent contextuality. The simplest (requiring fewest projectors) state-dependent non-contextuality inequalities were introduced in \cite{KCBS08}. Finally, there exist certain sets $\mathcal{S}$ exhibiting state-dependent contextuality of a special type. For these sets, it is possible to find valid $\{0,1\}$-colorings, however in any such coloring, there exist two special non-orthogonal vectors $|v_1\rangle$ and $|v_2\rangle$ that cannot both be assigned the value $1$. This class of special statistical state-dependent KS arguments were introduced by Clifton in \cite{Clifton93}.


An interesting question is to study how contextual quantum mechanics is, when considering such state-independent and state-dependent non-contextuality argument. In \cite{ACC15}, the authors investigated how large contextuality can be in quantum theory for state-dependent non-contextuality inequalities. In these inequalities, the contextuality witness is expressed as a sum $S$ of $n$ probabilities (corresponding to $n$ projectors). The independence number $\alpha(G)$ and the Lovasz-theta number $\theta(G)$ of the corresponsing orthogonality graph are then the maximum values attainable in non-contextual theories, and in quantum theory respectively. The authors made use of a graph-theoretic result by Feige \cite{Feige1} stating that for every $\epsilon > 0$, an $n$-vertex graph $G$ exists such that $\theta(G)/\alpha(G) > n^{1-\epsilon}$ to show that quantum theory allows for maximal contextuality in this case. In \cite{Simmons17}, Simmons studied how contextual quantum mechanics is when considering state-independent non-contextuality inequalities, as measured by the fraction of projectors in the KS proof that must have a valuation that depends on the \emph{context} in which they are measured. 

This paper is organized as follows. We first introduce some notation and elementary concepts, in particular the representation of KS sets as graphs. In the next section, we introduce KS sets with large violations, making use of a connection between Kochen Specker proofs and pseudo-telepathy games from \cite{RW04}. We present an application of the large violation KS sets to finding classical channels where shared entanglement increases the one-shot zero-error capacity. We also show that, despite their foundational interest, the maximum violation of KS state-independent non-contextuality inequalities do not serve to certify intrinsic randomness. To remedy this fault, we study large violations in a class of state-dependent non-contextuality inequalities known as $01$-gadgets in the subsequent section \ref{sec:large-gadget}. Interestingly, we show a technical result here that the minimum dimension of a faithful orthogonal representation of a graph in $\mathbb{R}^d$ is not graph monotone, a result that has fundamental applications in constructions of KS proofs. We derive a one-to-one correspondence between $01$-gadgets in $\mathbb{C}^d$ and two-player Hardy paradoxes for the maximally entangled state in $\mathbb{C}^d \otimes \mathbb{C}^d$. We use this connection to construct novel large violation $01$-gadgets, with an explicit construction shown in Appendix II. Finally, we show that $01$-gadgets are natural candidates for self-testing under the assumption of a fixed dimension, a result which has interesting application for contextuality-based randomness generation. We conclude with some open questions.

\section{Preliminaries}\label{sec:prel}
In this section, we define some graph-theoretic notation useful in the study of KS contextuality.


\paragraph*{Orthogonality Graphs.}
Throughout the paper, we will deal with simple undirected finite graphs $G$, i.e., finite graphs without loops, multi-edges or directed edges. We denote $V(G)$ the vertices of $G$ and $E(G)$ the edges of $G$. 
It is convenient to represent the orthogonality relations in a KS set $\mathcal{S}$ by means of an orthogonality graph $G_{\mathcal{S}}$ \cite{CSW10, CSW14}. In the orthogonality graph, each vector $|v_i\rangle$ in $\mathcal{S}$ is represented by a vertex $v_i$ of $ G_{\mathcal{S}}$ and two vertices $v_1, v_2$ of $G_{\mathcal{S}}$ are connected by an edge if the associated vectors $|v_1 \rangle, |v_2 \rangle$ are orthogonal, i.e. $v_1\sim v_2$ if $\langle v_1 | v_2 \rangle = 0$.

\paragraph*{Orthogonal representations.}
For a given graph $G$, an orthogonal representation $\mathcal{S}$ of $G$ in dimension $d$ is a set of non-zero vectors $\mathcal{S}=\{|v_i \rangle\}$ in $\mathbb{C}^d$ obeying the orthogonality conditions imposed by the edges of the graph, i.e., $v_1 \sim v_2 \Rightarrow \langle v_1|v_2 \rangle=0$ \cite{Lovasz87}. We denote by $d(G)$ the minimum dimension of an orthogonal representation of $G$ and we say that $G$ has dimension $d(G)$. A \emph{faithful} orthogonal representation of $G$ is given by a set of vectors $\mathcal{S}=\{|v_i \rangle\}$ that in addition obey the condition that non-adjacent vertices are assigned non-orthogonal vectors, i.e., $v_1 \sim v_2 \Leftrightarrow  \langle v_1|v_2 \rangle=0$ and that distinct vertices are assigned different vectors, i.e., $v_1 \neq v_2 \Leftrightarrow |v_1 \rangle \neq |v_2 \rangle$.
We denote by $d^*(G)$ the minimum dimension of such a faithful orthogonal representation of $G$ and we say that $G$ has faithful dimension $d^*(G)$.


\paragraph*{KS graphs.}
While the non-$\{0,1\}$-colorability of a set $S$ translates into the non-$\{0,1\}$-colorability of its orthogonality graph $G_\mathcal{S}$, the non-$\{0,1\}$-colorability of an arbitrary graph $G$ translates into the non-$\{0,1\}$-colorability of one of its orthogonal representations only if this representation has the minimal dimension $d(G)=\omega(G)$, where $\omega(G)$ denotes the size of the maximum clique in the graph. 
If a graph $G$ is not $\{0,1\}$-colorable and has dimension $d(G)=\omega(G)$, it  follows that its minimal orthogonal representation $\mathcal{S}$ forms a KS set. If in addition $d^*(G)=\omega(G)$, we say that $G$ is a KS graph (this last condition can always be obtained by considering the faithful version of $G$, i.e., the orthogonality graph $G_\mathcal{S}$ of its minimal orthogonal representation $\mathcal{S}$).


The problem of finding KS sets can thus be reduced to the problem of finding KS graphs. However, deciding if a graph is $\{0,1\}$-colorable is known to be NP-complete \cite{Arends09}. In addition, while finding an orthogonal representation for a given graph can be expressed as finding a solution to a system of polynomial equations, efficient numerical methods for finding such representations are still lacking. Thus, finding KS sets in arbitrary dimensions is a difficult problem towards which a huge amount of effort has been expended \cite{CA96}. In particular, ``records'' of minimal Kochen-Specker systems in different dimensions have been studied \cite{CEG96}, the minimal KS system in dimension four is the $18$-vector system due to Cabello et al. \cite{CEG96, Cab08} while lower bounds on the size of minimal KS systems in other dimensions have also been established.

\section{Large violations of KS state-independent non-contextuality inequalities}
\label{sec:Large-KS}
In \cite{ACC15}, the authors investigated how large contextuality can be in quantum theory for state-dependent non-contextuality inequalities. In these inequalities, the contextuality witness is expressed as a sum $S$ of $n$ probabilities (corresponding to $n$ projectors). The independence number $\alpha(G)$ and the Lovasz-theta number $\theta(G)$ of the corresponsing orthogonality graph are then the maximum values attainable in non-contextual theories, and in quantum theory respectively. The authors made use of a graph-theoretic result by Feige \cite{Feige1} stating that for every $\epsilon > 0$, an $n$-vertex graph $G$ exists such that $\theta(G)/\alpha(G) > n^{1-\epsilon}$ to show that quantum theory allows for maximal contextuality in this case. It must be noted, however that the proof is not constructive and does not single out explicit scenarios.

An open question has remained as to how contextual quantum mechanics is, when considering state-independent non-contextuality inequalities, i.e., inequalities arising from KS proofs or from statistical state-independent KS arguments. In \cite{Simmons17}, Simmons carried out mathematical investigations as to how (maximally) contextual is quantum mechanics, measuring this by the fraction of projectors in the KS proof that must have a valuation that depends on the \emph{context} in which they are measured. An upper bound was derived on this quantity $q(G)$ for arbitrary KS proofs in $n$-dimensional Hilbert space as \cite{Simmons17}
\begin{eqnarray}
q(G) &\leq & \left(1 - \frac{1}{n} \right)^{n-1} - \frac{1}{2^{n-1}} \nonumber \\
 &\leq & 4251920575/11019960576 \approx 0.385838. 
\end{eqnarray}
However, it was shown that there is a gulf between the values achieved by the best-known KS proofs and the upper bound shown above. The Peres-Mermin magic square \cite{Peres, Mermin} achieves a value $q(G) = 0.08\dot{3}$, Cabello's 18-vertex proof \cite{CEG96, Cab08} achieves the value $q(G) = 0.0\dot{5}$, while two-qubit stabiliser quantum mechanics (that forms a KS proof using $60$ projectors and $105$ contexts) achieves the value $q(G) = 0.1$. Note that for the Peres-Mermin magic square, the value for the alternative figure-of-merit is $\theta(G)/\alpha(G) = 6/5 = 1.2$, while the value for Cabello's 18-vertex proof is $\theta(G)/\alpha(G) = 4.5/4 = 1.125$. Here, we show a KS proof for Hilbert spaces of dimension $n \geq 2^{17}$ which achieves $\theta(G)/\alpha(G) \approx 2$. Attempts to find KS proofs that yield a high value of this quantity can be thought of as extensions of the search for small Kochen-Specker sets, to which much attention has been devoted in the literature \cite{CA96, CEG96, Cab08}. Finding KS proofs with large violation is thus an open question of fundamental significance, besides having applications as we shall see.  In this paper, we consider the usual figure-of-merit, namely the ratio of the quantum to the classical bias, and show a KS set with large value for this parameter. More precisely, we show that the corresponding ratio of the quantum bias to the classical bias (defined as the ratio of the quantities normalized such that $\theta(G) = 1$ minus the random assignment value of $1/2$) is $O\left(\sqrt{n}/\log n \right)$, yielding the large violation statement.

\begin{theorem}
In Hilbert spaces of dimension $n \geq 2^{17}$, there exist Kochen Specker proofs consisting of $N = 2^{n-1} + n^2/2$ projectors such that for the corresponding orthogonality graph $G$, it holds that
\begin{eqnarray}
\theta(G) &=& \frac{2^{n-1}}{n} + \frac{n}{2}, \nonumber \\
\alpha(G) &\leq& \left(\frac{1}{2} + \frac{8 \sqrt{n} \log n}{n-1} \right) \left( \frac{2^{n-1}}{n} + \frac{n}{2} \right).
\end{eqnarray}
\end{theorem}
\begin{proof}
The proof follows from two crucial ingredients. Firstly, we make use of the hidden-matching game from Buhrman et al. \cite{BSW10} which is non-locality game between two players that can be won with certainty by a quantum strategy using $\log n$ shared EPR-pairs between the players, while any classical strategy has winning probability at most $\frac{1}{2} + O\left(\frac{\log n}{\sqrt{n}} \right)$. Secondly, we make use of a connection between pseudo-telepathy games and weak Kochen-Specker proofs shown by Renner and Wolf in \cite{RW04}. 

The non-local hidden matching game is defined as follows \cite{BSW10}. Let $n$ be a power of $2$ and $M_n$ be the set of all perfect matchings on the set $[n]$. Alice is given as input $x \in \{0,1\}^n$ and Bob is given as input $M \in M_n$, distributed uniformly. Alice’s output is a string $a \in \{0,1\}^{\log n}$ and Bob’s output is an edge $(i, j) \in M$ and a bit string $b \in \{0, 1\}^{\log n}$. They win the game if the following condition is true $(a \oplus b)(i \oplus j) = x_i \oplus x_j$. In \cite{BSW10}, the authors showed that the classical value of the game obeys $\omega_c \leq \frac{1}{2} + O \left( \frac{\log n}{\sqrt{n}} \right) $, while a quantum strategy exists that achives the value $\omega_q = 1$. In other words, the non-local hidden matching game is a pseudo-telepathy game. In \cite{RW04}, Renner and Wolf showed that for every two-player pseudo-telepathy game with the optimal strategy achieved by a maximally entangled state, the union of the sets of projectors measured by the two parties constitutes a (weak) KS proof. 

We first derive the optimal constant in the big-O value for the classical success probability in the game. 
\begin{lemma}
The classical success probability in the non-local hidden matching game is bounded as
\begin{eqnarray}
\omega_c \leq  \left(\frac{1}{2} + \frac{8 \sqrt{n} \log n}{n-1} \right).
\end{eqnarray}
\end{lemma}
\begin{proof}
We follow the argument in \cite{BSW10} using the Kahn-Kalai-Linial (KKL) inequality. The inequality states that for every $\delta \in [0,1]$ and $f: \{0,1\}^n \rightarrow \{-1,0,1\}$ we have
\begin{eqnarray}
\sum_S \delta^{|S|} \hat{f}(S)^2 \leq \left(\text{Pr}[f \neq 0] \right)^{2/(1+ \delta)},
\end{eqnarray}
where $\hat{f}(S)$ denotes the Fourier coefficient. Informally, the inequality says that a $\{-1,0,1\}$-valued function with small support cannot have too much of its Fourier weight on low degrees. The KKL inequality is used to bound the expected bias of $k$-bit parities over a set $A \subseteq \{0,1\}^n$.
Suppose we pick a set $S \subseteq [n]$ of $k$ indices uniformly at random, and consider the parity of the $k$-bit substring induced by $S$ and a uniformly random $x \in A$. 
Intuitively, if $A$ is large then we expect that for most $S$, the bias $\beta_S = \mathbb{E}_{x \in A} \left[\chi_S(x) \right]$ of this parity to be small: the number of $x \in A$ with $\chi_S(x) = 1$ should be roughly the same as with $\chi_S(x) = -1$.
From \cite{Wol08}, we derive 
\begin{eqnarray}
\label{eq:Sbinom}
\sum_{S \in\binom{[n]}{k}} \beta_{s}^{2} \leq \frac{1}{\delta^{k}}\left(\frac{2^{n}}{|A|}\right)^{2 \delta},
\end{eqnarray}
and
\begin{eqnarray}
\label{eq:expbetasq}
E_{s}\left[\beta_{s}^{2}\right]=\frac{1}{\binom{n}{k}} \sum_{S \in\binom{[n]}{k}} \beta_{s}^{2}=O\left(\frac{\log _{2}\left(\frac{2^{n}}{|A|}\right)}{n}\right)^{k},
\end{eqnarray}
where $\delta=\frac{k}{\left[2 \ln \left(\frac{2^{n}}{|A|}\right)\right]}$ minimizes the right hand side of \eqref{eq:Sbinom}. So that
\begin{eqnarray}
\label{eq:sumbetasq}
\sum_{S \in\binom{[n]}{k}} \beta_{s}^{2} &\leq & \frac{1}{\left\{\frac{k}{\left[2 \ln \left(\frac{2^{n}}{|A|}\right)\right]}\right\}^{k}} \left(\frac{2^{n}}{|A|}\right)^{\frac{2k}{\left[2 \ln \left(\frac{2^{n}}{|A|}\right)\right]}} \nonumber \\
&\leq &
\frac{1}{\left\{\frac{k}{\left[2 \log_2 \left(\frac{2^{n}}{|A|}\right)\right]}\right\}^{k}} \left(\frac{2^{n}}{|A|}\right)^{\frac{2k}{\left[2 \log_2 \left(\frac{2^{n}}{|A|}\right)\right]}}.
\end{eqnarray}

For the hidden matching problem we have $k=2$. If Alice sends $c$ bits to Bob, $\frac{|A|}{2^{n}}=2^{-c}$.

From \eqref{eq:sumbetasq}, we have 
\begin{eqnarray}
\label{eq:normsumbeta}
\frac{1}{\binom{n}{2}}\sum_{S \in\binom{[n]}{2}} \beta_{s}^{2} & \leq &
\frac{1}{\binom{n}{2}}\frac{1}{\left\{\frac{2}{\left[2 \log_2 \left(2^c\right)\right]}\right\}^{2}} \left(2^c\right)^{\frac{2\times 2}{\left[2 \log_2 \left(2^c\right)\right]}} \nonumber \\
&=& \frac{1}{\frac{n(n-1)}{2}} \frac{1}{\left(\frac{2}{2 c}\right)^{2}}\left(2^{c}\right)^{\frac{2 \times 2}{2 c}} \nonumber \\
&=&\frac{8 c^{2}}{n(n-1)}.
\end{eqnarray}

From \eqref{eq:expbetasq}, we have 
\begin{eqnarray}
\label{eq:bigObeta}
\frac{1}{\binom{n}{2}} \sum_{S \in\binom{[n]}{2}} \beta_{s}^{2}=O\left(\frac{\log _{2}\left(2^c\right)}{n}\right)^{2}=O\left(\frac{c}{n}\right)^2.
\end{eqnarray}

Based on \eqref{eq:normsumbeta} and \eqref{eq:bigObeta}, we deduce that the constant $\alpha$ obeys
\begin{eqnarray}
\alpha\geq\frac{8 n}{n-1}.
\end{eqnarray}
Thus, we have $\omega_c \leq  \left(\frac{1}{2} + \frac{8 \sqrt{n} \log n}{n-1} \right)$.
\end{proof}

Now, following the quantum protocol of the non-local hidden matching problem, we can construct a Kochen-Specker set. The orthogonal bases pair of $\mathcal{H}=\mathcal{H}_{n} \otimes \mathcal{H}_{n} $is $(P_x,Q_M)$, where the $x$ is the input of Alice and $M$ is the input of Bob. The set of projectors measured by Alice is $P = \cup_x P_x$ and the corresponding set for Bob is $Q = \cup_M Q_M$. The KS vector set is then given as $S = P \cup Q$.

Let $H^{\prime}=H^{\otimes \log n}$, the vector $P_x^i$ in basis $P_x$ is given as:
$$
P_x^i=U_{x} H^{\otimes \log n}|i\rangle=U_{x}\left(\begin{array}{c}
H_{1 i}^{\prime} \\
H_{2 i}^{\prime} \\
\vdots \\
H_{n i}^{\prime}
\end{array}\right)=
\left(\begin{array}{c}
(-1)^{x_{1}+1 \cdot i} \\
(-1)^{x_{2}+2 \cdot i} \\
\vdots \\
(-1)^{x_{n}+n \cdot i}
\end{array}\right)
$$
where $i\in \{0,1\}^{\log n}$  and $U_{x}=\left(\begin{array}{llll}
(-1)^{x_{1}} & & & \\
& (-1)^{x_{2}} & & \\
& & \ddots & \\
& & & (-1)^{x_{n}}
\end{array}\right)$ is an $x$-dependent phase-flip matrix, and $x_i$ is the $i-th$ bit of the binary bit string $x$. For different vectors $P_x^i,P_x^j$ ($i\neq j \in \{0,1\}^{\log n}$) in the basis $P_x$ we have 
$$\langle P_x^i|P_x^j\rangle=\left(\begin{array}{c}
(-1)^{x_{1}} H_{1 i}^{\prime} \\
(-1)^{x_{2}} H_{2 i}^{\prime} \\
\vdots \\
(-1)^{x_{n}} H_{n i}^{\prime}
\end{array}\right)^{T}\left(\begin{array}{c}
(-1)^{x_{1}} H_{1 j}^{\prime} \\
(-1)^{x_{2}} H_{2 j}^{\prime} \\
\vdots \\
(-1)^{x_{n}} H_{n j}^{\prime}
\end{array}\right)$$
$$=(-1)^{2 x_{1}} H_{1 i}^{\prime} H_{1 j}^{\prime}+(-1)^{2 x_{2}} H_{2 i}^{\prime} H_{2 j}^{\prime}+\cdots+(-1)^{2 x_{n}} H_{n i}^{\prime} H_{n j}^{\prime} $$
$$=
H_{1 i}^{\prime} H_{1 j}^{\prime}+H_{2 i}^{\prime} H_{2 j}^{\prime}+\cdots+H_{n i}^{\prime} H_{n j}^{\prime}=0.
$$
Now $P_x$ is a complete basis with $n$ orthogonal vectors and
$P=\bigcup_{x} P_{x}$.


The vector $Q_M^{k,(i,j)}$ in basis $Q_M$ is given as:
$$
Q_{M}^{k,(i, j)}=P_{i j} H^{\otimes \log n}|k\rangle=P_{i j}\left(\begin{array}{c}
H_{1 k}^{\prime} \\
H_{2 k}^{\prime} \\
\vdots \\
H_{n k}^{\prime}
\end{array}\right)=\left(\begin{array}{c}
0 \\
\vdots \\
H_{i k}^{\prime} \\
0 \\
\vdots \\
H_{j k}^{\prime} \\
0\\
\vdots \\
\end{array}\right)
$$
where $k\in \{0,1\}^{\log n}$ and $(i,j)$ is a disjoint pair of the perfect matching $M$ and $P_{i,j}$ is a diagonal matrix with $(i,i)$ and $(j,j)$-th entries equal to $1$, and the rest of the entries being $0$.
Since the elements in $H^{\prime}=H^{\otimes \log n}$ must be $1$ or $-1$, given the determined matching $M$ and a pair $(i,j)$ with different $k$, there are just two distinct vectors: $Q_M^{(i,j)}$ with a $1$ at positions $i$ and $j$ and $0$ elsewhere, and $Q_{M}^{'(i,j)}$ with $1$ at position $i$, $-1$ at position $j$ and $0$ elsewhere.


Given a determined perfect matching $M$, the $\frac{n}{2}$ pairs $(i,j)$ are all disjoint in the matching, i.e., if $(i,j)$ and $(i',j')$ are pairs in $M$, then $i\neq j\neq i'\neq j'$. We obtain that 
\begin{eqnarray}
\langle Q_M^{(i,j)}|Q_M^{(i',j')}\rangle=0, \; \; \langle Q_M^{(i,j)}|Q_M^{'(i',j')}\rangle=0, \nonumber \\
\langle Q_M^{'(i,j)}|Q_M^{(i',j')}\rangle=0, \; \; \langle Q_M^{'(i,j)}|Q_M^{'(i',j')}\rangle=0.
\end{eqnarray}
Now $Q_M$ is a complete basis with $n$ orthogonal vectors and
$Q=\bigcup_{M} Q_{M}$. The KS vector set is then given by $S= P\cup Q$.
\end{proof}
\subsection{Application I. Entanglement assisted one-shot zero-error capacity of a classical channel }

For a classical channel $\mathcal{N}$ connecting the the sender Alice and receiver Bob, the behaviour of the channel is described by the conditional probability distribution over outputs given the input. Two inputs of the channel $\mathcal{N}$ are confusable if their outputs overlap. Shannon introduced the confusability graph $G(\mathcal{N})$ of a classical channel $\mathcal{N}$: the inputs are expressed as vertices in the graph, two vertices are connected by an edge if and only if they are confusable. Classically, the maximum number of different messages Alice can send to Bob without error through the classical channel $\mathcal{N}$ is the independence number of the confusability graph $G(\mathcal{N})$. In \cite{CLMW10}, Cubitt et al. showed that given single use of a channel based on certain proofs of the KS theorem, entangled states of a system shared by the sender and receiver can be used to increase the number of (classical) messages which can be sent with no chance of error. In other words, for these KS channels, one can improve the zero-error classical communication capacity using entanglement. 

We now need to verify that the classical channel $\mathcal{N}$ constructed from our KS graph $G$ satisfies the conditions on the channel imposed by Cubitt et al.'s proof. If so, we can show an example of a classical channel $\mathcal{N}$ for which the entanglement-assisted classical zero-error capacity far exceeds the classical zero-error capacity without shared entanglement. Firstly, we construct the classical channel $\mathcal{N}$ using our KS graph $G$. The classical one-shot zero-error capacity of the corresponding channel is:
\begin{eqnarray}
c_{0}(\mathcal{N})=\alpha(G)=\frac{2^{(n-1)}+\frac{n^{2}}{2}}{n} \left(\frac{1}{2}+O\left(\frac{\log n}{\sqrt{n}}\right)\right).
\end{eqnarray}

Now, the condition that must be satisfied by the KS graph $G$ is that the graph can be partitioned into an integral number of disjoint cliques of size $n$. We show in Lemma \ref{lem:cliques} that our KS graph $G$ consists of $q=2^{(n-1)}+\frac{n^{2}}{2}$ disjoint cliques of size $n$.
This implies that the classical one-short zero-error capacity of the corresponding channel can be increased when the sender and receiver share a maximally entangled state $\rho_{AB}$ of rank $n$.
The entanglement-assisted zero-error capacity $c_{SE}$ can be shown to be exactly the number of the disjoint cliques:
\begin{eqnarray}
c_{SE}(\mathcal{N})=q=\frac{2^{(n-1)}+\frac{n^{2}}{2}}{n}.
\end{eqnarray}
To see this, note the protocol for the task outlined in the proof by Cubitt et al. \cite{CLMW10}. In order to send a message $m \in q$, Alice can choose a projector $j$ in clique $m$ randomly and measure her side of the state. She then inputs $(m,j)$ into the channel. Bob's output will be a subset containing $(m,j)$ and its orthogonal vertices. After performing a projective measurement on his side of the state, he can infer which message $m \in q$ Alice has sent with certainty. In other words, the protocol works perfectly when the KS graph $G$ can be partitioned into an integral number of disjoint cliques of size $n$.

\begin{lemma}
\label{lem:cliques}
The KS graph $G$ obtained from the non-local hidden matching problem can be partitioned into 
\begin{equation}
q=\frac{2^{(n-1)}}{n}+\frac{n}{2}.
\end{equation}
disjoint cliques of size $n$.
\end{lemma}
\begin{proof}
To show this, we show the equivalent statement that in the optimal quantum strategy for the non-local hidden matching game, the bases measured by the two parties do not have any overlap. In other words, for any pair of bit strings $x \neq x'$, if one of the vectors $P_x^i=\pm P_{x'}^j$, then the bases $P_x=P_{x'}$. This means that if one of the vectors $P_x^i=\pm P_{x'}^j$, then for any other vector $P_x^p$ in basis $P_x$, we can find a corresponding vector $P_{x'}^q$ in basis $P_{x'}$, such that $P_x^p=\pm P_{x'}^q$. Given $P_x^i=\pm P_{x'}^j$, we have 
that for $\forall s\in \{1,2,\cdots,n\}$, $(-1)^{x_{s}+s \cdot i}=\pm(-1)^{x'_{s}+s \cdot j}$.

This implies that for any other vector $P_x^p$ in basis $P_x$, we have 
\begin{eqnarray}
(-1)^{x_{s}+s \cdot p}&=&\pm (-1)^{x_{s}+s \cdot i+s\cdot d} \nonumber \\
&=&\pm (-1)^{x'_{s}+s \cdot j+s\cdot d} \nonumber \\
&=&\pm (-1)^{x'_{s}+s \cdot q}.
\end{eqnarray}
Writing $s,p,i,d,q$ as $\log n$-bit binary strings, we obtain that
\begin{eqnarray}
s\cdot p= \sum_{r=1}^{\log n} s_{r} \cdot p_{r}(\bmod 2)&=&\sum_{r=1}^{\log n} s_{r} \cdot i_{r}+\sum_{r=1}^{\log n} s_{r} \cdot d_{r}(\bmod 2) \nonumber \\
&=&\sum_{r=1}^{\log n} s_{r} \cdot(i_r+d_r)(\bmod 2)
\end{eqnarray}
so that
\begin{eqnarray}
p_r=(i_r+d_r)(\bmod 2)=i_r\oplus d_r.
\end{eqnarray}
Similarly, we have $q_r=j_r\oplus d_r=j_r\oplus p_r\oplus i_r$. We have thus found the exact vector $P_{x'}^q$ in basis $P_{x'}$ , such that $P_x^p=\pm P_{x'}^q$. Therefore, for given different bit strings $x$ and $x'$, the bases $P_{x}$ and $P_{x'}$ will never overlap.

Now, in the hidden matching problem, Bob's matching is chosen uniformly from $$M\in \left\{M_{k} \mid k \in\{0, \ldots, \frac{n}{2}-1\}\right\}$$ And the pairs $(i,j)$ in matching $M_{k}$ are given by:
\begin{eqnarray}
j=\frac{n}{2}+1+(i+k-1\bmod \frac{n}{2}), \; \; \; \; i \leq \frac{n}{2}. 
\end{eqnarray}
Therefore, the pair $(i,j)$ will never repeat in different matchings $M$ and $M'$, so that the bases $Q_M$ and $Q_{M'}$ never overlap. Having seen that the bases of Alice and Bob do not overlap, we infer that the total number of disjoint bases in the KS set $S$ is the total number of questions of Alice and Bob, which gives $q= \frac{2^{n-1}}{n}+\frac{n}{2}$. Since each basis includes $n$ orthogonal vectors, we also deduce that the total number of vectors in the KS set $S$ is $2^{n-1}+\frac{n^2}{2}$.




%

\end{proof}

\subsection{Application II. KS State-Independent Contextuality does not certify intrinsic randomness}
The result in the previous subsection has foundational significance and has some very interesting applications. Surprisingly, an application that one may anticipate from every contextual (and non-local) behavior, namely the certification of randomness cannot be obtained from the violation of any KS state-independent non-contextuality inequality. This comes from the following curious observation. 

Consider for the sake of concreteness the KS proof known as the Peres-Mermin (PM) square, consisting of nine binary observables $\{\sigma_x \otimes \mathds{1}, \mathds{1} \otimes \sigma_x, \sigma_x \otimes \sigma_x, \mathds{1} \otimes \sigma_z, \sigma_z \otimes \mathds{1}, \sigma_z \otimes \sigma_z, \sigma_x \otimes \sigma_z, \sigma_z \otimes \sigma_x, \sigma_y \otimes \sigma_y \}$, where $\sigma_x, \sigma_y \sigma_y$ refer to the usual Pauli observables. The commutation hypergraph of these observables consists of three rows and three columns, with the quantum mechanical predictions 
\begin{eqnarray}
\langle \left(\sigma_x \otimes \mathds{1} \right) \left( \mathds{1} \otimes \sigma_x  \right) \left( \sigma_x \otimes \sigma_x \right) \rangle &=& 1, \nonumber \\
\langle \left(\mathds{1} \otimes \sigma_z \right) \left( \sigma_z \otimes \mathds{1} \right) \left(\sigma_z \otimes \sigma_z \right) \rangle &=& 1, \nonumber \\
 \langle \left( \sigma_x \otimes \sigma_z \right) \left(\sigma_z \otimes \sigma_x  \right)  \left(  \sigma_y \otimes \sigma_y   \right)  \rangle &=& 1, \nonumber \\
\langle \left(\sigma_x \otimes \mathds{1}   \right) \left(\mathds{1} \otimes \sigma_z \right) \left( \sigma_x \otimes \sigma_z \right) \rangle &=& 1, \nonumber \\
\langle  \left( \mathds{1} \otimes \sigma_x  \right) \left( \sigma_z \otimes \mathds{1} \right)  \left(\sigma_z \otimes \sigma_x  \right)  \rangle &=& 1, \nonumber \\ 
\langle  \left( \sigma_x \otimes \sigma_x \right) \left(\sigma_z \otimes \sigma_z \right)  \left(  \sigma_y \otimes \sigma_y   \right) \rangle &=& -1.
\end{eqnarray}
On the other hand, denoting the observables in general as $\{A, B, C, a, b, c, \alpha, \beta, \gamma \}$ the maximum value in non-contextual theories of the expression
\begin{eqnarray}
\langle A B C \rangle + \langle a b c \rangle + \langle \alpha \beta \gamma \rangle + \langle A a \alpha \rangle + \langle B b \beta \rangle - \langle C c \gamma \rangle 
\end{eqnarray}
is $4$. Since non-contextual theories comprise (mixtures of) all deterministic behaviors, one may expect to certify randomness when a honest user observes the value of the expression to be equal to $6$ as predicted by quantum mechanics. 

However, this is not the case under the general paradigm of randomness certification \cite{}. In this paradigm, the observable(s) from which the honest user intends to extract the randomness (the hashing function $h$) is announced beforehand and is assumed to be known to the adversary, the reason being that no a priori private randomness is available to randomise the choice of hashing function. Now, the maximum quantum value of the state-independent non-contextuality inequality (the value $6$ in the example of the PM square) is achievable by the maximally mixed state (the state $\frac{\mathds{1}}{2} \otimes \frac{\mathds{1}}{2}$ in the PM square). These facts together imply that, irrespective of the choice of hashing function $h$ used by the honest party, a (contextual) behavior can be found that achieves the same maximum value of the inequality while also having the property that the output of the hashing function is deterministic. Indeed, this contextual behavior can be obtained by simply performing the measurements on the eigenstate of the hash observable. In the concrete example of the PM square, if the honest user chooses to extract the randomness from the observable $\sigma_x \otimes \sigma_x$, then there exists a contextual behavior (obtained by performing the PM measurements on the state $| + \rangle \otimes |+ \rangle$ where $|+ \rangle = \frac{1}{\sqrt{2}} \left( | 0 \rangle + | 1 \rangle \right)$) that achieves the value $6$ for the PM expression and such that the outcome of measurement of $\sigma_x \otimes \sigma_x$ is deterministic. It can be readily seen that the same observation extends to any arbitrary state-independent non-contextuality inequality, so that KS state-independent contextuality does not certify any randomness. Indeed, the state-independent inequalities are as yet unique in this respect. It would be interesting to see if there is any other non-contextuality or Bell inequality with the property that its violation does not certify any randomness (even against a quantum adversary as considered here).

\section{Large violations in $01$-gadgets and applications.}
\label{sec:large-gadget}
In contrast to state-independent non-contextuality inequalities, state-dependent inequalities in general allow for certification of intrinsic quantum randomness, for instance see the protocols in \cite{RBHH+15, BRGH+16, WBGH+16}. In particular, the state-dependent inequalities from measurement scenarios known as $01$-gadgets \cite{RRHP+20} are especially useful in 'localising' the value-indefiniteness that is guaranteed by the violation of a non-contextuality inequality \cite{ACS15, ACCS12}. The $01$-gadgets are $\{0,1\}$-colorable and thus do not represent by themselves KS sets. However, they do not admit arbitrary $\{0,1\}$-coloring: in any $\{0,1\}$-coloring of a $01$-gadget, there exist two special non-orthogonal vectors $|v_1\rangle$ and $|v_2\rangle$ that cannot both be assigned the value $1$. 

The $01$-gadgets or 'bugs' were first introduced as a means of constructing statistical KS arguments by Clifton \cite{Clifton93} (see Fig. \ref{fig:Clifton}) and have since been studied in the literature. In particular, $01$-gadgets were also used in \cite{Arends09} to show that the problem of checking whether certain families of graphs (which represent natural candidates for KS sets) are $\{0,1\}$-colorable is NP-complete. Specific $01$-gadgets have already been studied in the literature, for instance as 'definite prediction sets' in \cite{CA96} and recently as 'true-implies-false sets' in \cite{APSS18} where also minimal constructions in several dimensions were explored. 
Recently, some of us showed that $01$-gadgets form a fundamental primitive in constructing KS proofs, in the sense that every KS set contains a $01$-gadget and from every $01$-gadget one can construct a KS set.

\begin{figure}[t] 
\label{fig:Clifton}
\centering
\includegraphics[width=10cm]{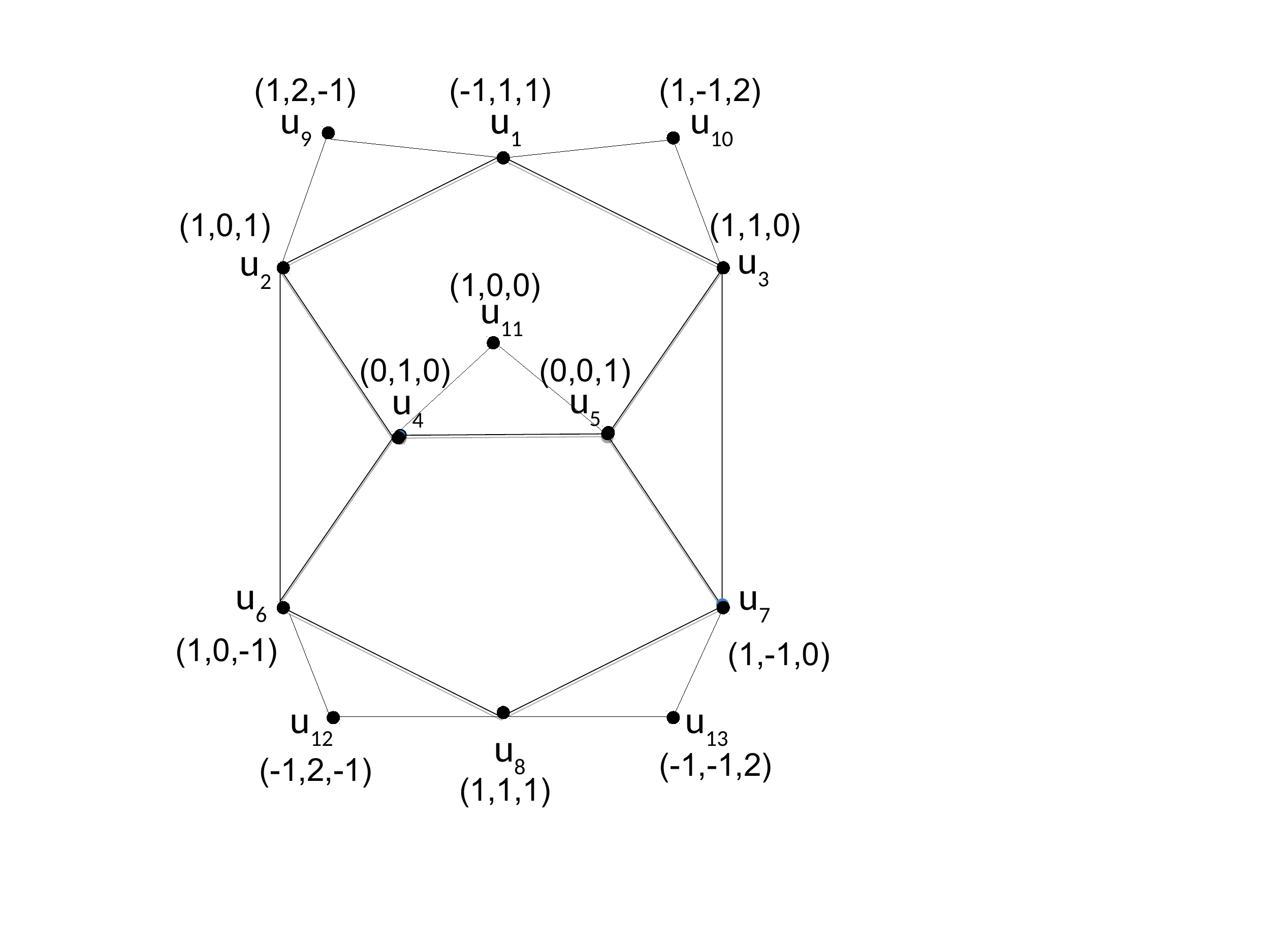}
 \caption{The original Clifton bug introduced in \cite{Clifton93}.}
\end{figure}

%

\begin{dfn}
	A $01$-gadget in dimension $d$ is a $\{0,1\}$-colorable set $\mathcal{S}_\text{gad}\subset\mathbb{C}^d$ of vectors containing two distinguished vectors $|v_1\rangle$ and $|v_2\rangle$ that are non-orthogonal, but for which $f(|v_1\rangle)+f(|v_2\rangle)\leq 1$ in every $\{0,1\}$-coloring $f$ of $\mathcal{S}_\text{gad}$.
\end{dfn}
In other words, while a $01$-gadget $\mathcal{S}_\text{gad}$ admits a $\{0,1\}$-coloring, in any such coloring the two distinguished non-orthogonal vertices cannot both be assigned the value $1$ (as if they were actually orthogonal).
We can give an equivalent, alternative definition of a gadget as a graph.
\begin{dfn}
	A $01$-gadget in dimension $d$ is a $\{0,1\}$-colorable graph $G_\text{gad}$ with faithful dimension $d^*(G_\text{gad})=\omega(G_\text{gad})=d$ and with two distinguished non-adjacent vertices $v_1 \nsim v_2$ such that $f(v_1)+f(v_2)\leq 1$ in every $\{0,1\}$-coloring $f$ of $G_\text{gad}$.
\end{dfn}
In the following when we refer to a $01$-gadget, we freely alternate between the equivalent set or graph definitions.

An example of a $01$-gadget in dimension 3 is given by the following set of 8 vectors in $\mathbb{C}^3$:
\begin{eqnarray}
	\label{eq:Clif-orth-rep}
&&	| u_1 \rangle = \frac{1}{\sqrt{3}}(-1,1,1), \; \; |u_2 \rangle = \frac{1}{\sqrt{2}}(1,1,0), \nonumber \\ 
&& |u_3 \rangle = \frac{1}{\sqrt{2}} (0,1,-1), |u_4 \rangle = (0,0,1), \nonumber \\
&&	|u_5 \rangle = (1,0,0), \; \; |u_6 \rangle = \frac{1}{\sqrt{2}}(1,-1,0), \nonumber \\
&& |u_7 \rangle = \frac{1}{\sqrt{2}}(0,1,1), \; \; |u_8 \rangle = \frac{1}{\sqrt{3}}(1,1,1), 
	\end{eqnarray}
where the two distinguished vectors are $|v_1\rangle=|u_1\rangle$ and $|v_2\rangle=|u_8\rangle$. Its  orthogonality graph is represented in Fig.~\ref{fig:Clifton}. It is easily seen from this graph representation that the vertices $u_1$ and $u_8$ cannot both be assigned the value 1, as this then necessarily leads to the adjacent vertices $u_4$ and $u_5$ to be both assigned the value 1, in contradiction with the $\{0,1\}$-coloring rules. This graph was identified by Clifton, following work by Stairs \cite{Clifton93, Stairs}, and used by him to construct statistical proofs of the Kochen-Specker theorem. We will refer to it as the Clifton gadget $G_{\text{Clif}}$. The Clifton gadget and similar gadgets were termed ``definite prediction sets" in \cite{CA96}. 


In this section, we show constructions of $01$-gadgets that achieve large violations in the sense that the overlap between the special vectors $|v_1\rangle$ and $|v_2\rangle$ can be made arbitrary. We also show constructions that achieve the self-testing property that under the constraint of a fixed dimension, there is a unique orthogonal realization (up to rotations) of the $01$-gadget. We apply our constructions to show interesting novel applications of the $01$-gadgets. We also show a property of the faithful representations that underscores how difficult it is to construct novel KS proofs and $01$-gadgets.

\subsection{The minimum dimension of a faithful orthogonal representation in $\mathbb{R}^d$ is not graph monotone}

In graph theory, a graph property$P$ is said to be monotone if every subgraph of a graph with property $P$ also has property $P$. In other words, the graph property is closed under removal of edges and vertices. Recall that a faithful orthogonal (also orthonormal, since all vectors are taken to have unit norm) representation of $G$ is given by a set of vectors $S = \{ | v_i \rangle \}$ in $\mathbb{C}^d$ that obey the orthogonality conditions imposed by the edges of the graph, and in addition obey the condition that non-adjacent vertices are assigned non-orthogonal vectors and that distinct vertices are assigned different vectors. We had denoted by $d^*(G)$ the minimum dimension of such a faithful orthogonal representation of $G$. Let us now denote by $d_R^*(G)$ the corresponding minimum dimension when $\mathbb{C}^d$ is replaced by $\mathbb{R}^d$ in the above definition. That is, $d_R^*(G)$ is the minimum dimension in real vector spaces of a faithful orthogonal representation of the graph $G$.

Given a graph $G$ that has a faithful orthogonal representation in dimension $d$, let us form a new graph $G \cup uv$ by adding an edge $uv$ (if such is possible) to $G$. In general, we expect that the minimum dimension of the orthogonal representation increases by this operation of adding edges, i.e., that $d_R^*(G \cup uv) > d_R^*(G)$. Conversely, consider the operation of deleting an edge $uv$ (if such is possible) of $G$. In general, we expect $d_R^*(G \setminus uv) \leq d_R^*(G)$. The question we address in this section is whether this holds always, i.e.,  
\begin{itemize}
\item Is the graph-property $P^{d,n}$ of all the graphs on $n$ vertices which admit a faithful orthogonal representation in $\mathbb{R}^d$ monotone-decreasing?
\end{itemize}
Surprisingly, we show that the answer to this question is negative. Our proof is constructive, we give an explicit example of a graph $G$ which has a faithful orthogonal representation in $\mathbb{R}^3$, and yet deleting an edge $uv \in E(G)$ increases the minimum dimension of the faithful orthogonal representation of the resulting graph, i.e., $d_R^*(G \setminus uv) > 3$. 

Given that one may readily discover candidate non-$\{0,1\}$-colorable graphs or candidate $01$-gadget graphs, see for instance \cite{AM78, AM80}, a large part of the difficulty in constructing KS proofs and $01$-gadgets is in finding minimum dimensional faithful orthogonal representations of such candidate graphs. The surprising property that deleting edges does not retain the dimension of the representation and that in some cases one may have to increase the dimension of the representation after this operation, is thus an important discovery in the research project aimed at constructing minimal KS proofs and $01$-gadgets. It also underscores the importance of the alternative construction methods presented in the rest of the paper.

\begin{figure}[t] 
\label{fig:gadg-sixty}
\centering
\includegraphics[width=7cm]{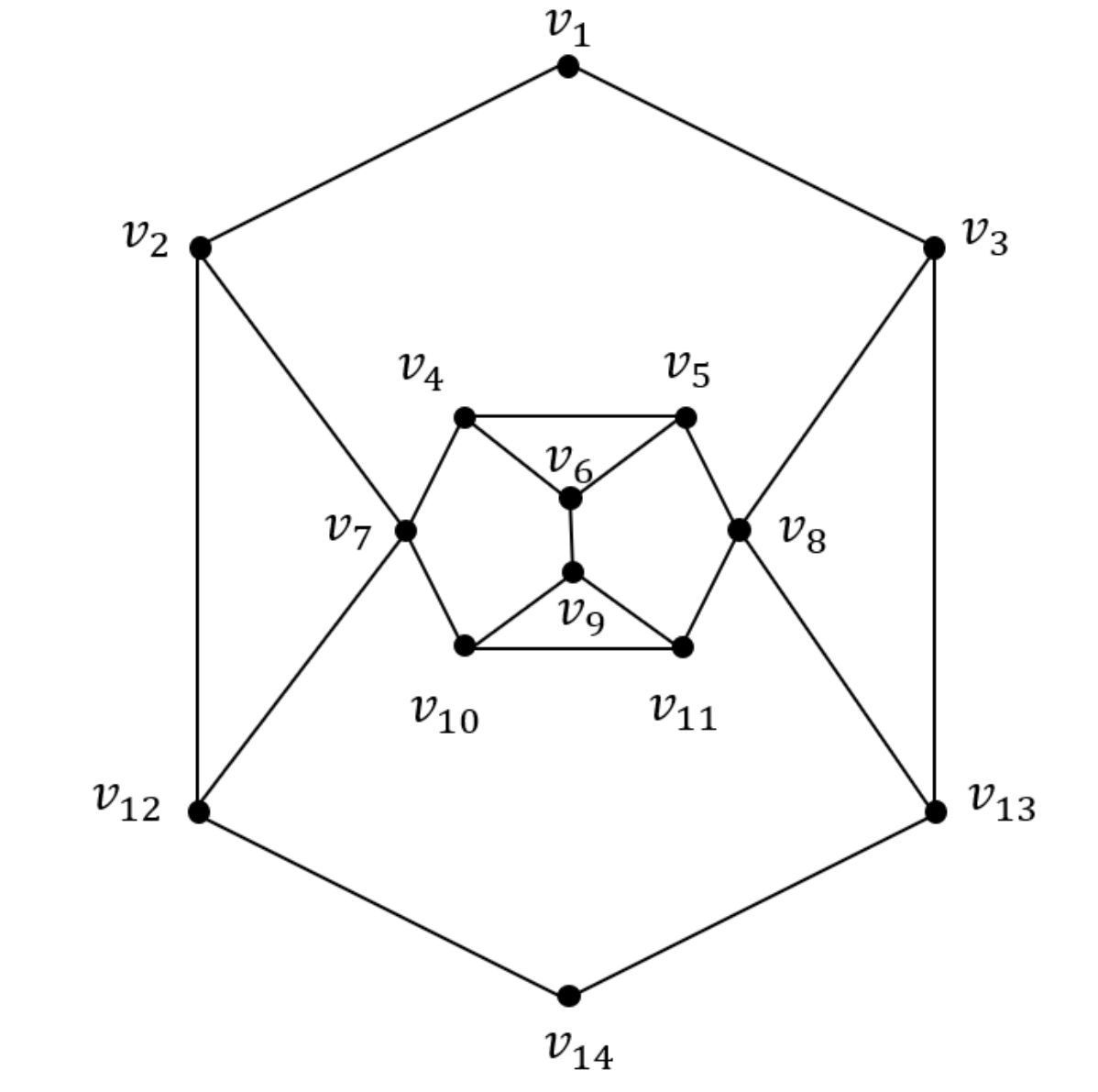}
 \caption{This gadget has representation in $\mathbb{R}^3$, and moreover the two distinguished vectors $|v_1 \rangle$ and $|v_2 \rangle$ have overlap at most $1/2$ in this dimension, i.e., $|\langle v_1 | v_2 \rangle| \leq 1/2$. Specifically, an optimal orthogonal representation is given by $\langle v_1 | = (-3\sqrt{2}, -3, 3)$,  $\langle v_2 | = (0,1,1)$, $\langle v_3 | = (-6 \sqrt{2}, 3, -9)$, $\langle v_4 |= (0, \sqrt{3}, -1)$, $\langle v_5 | = (-2\sqrt{2}, 1, \sqrt{3})$, $\langle v_6| = (-\sqrt{2}, -1, -\sqrt{3})$, $\langle v_7 | = (1,0,0)$, $\langle v_8 |= (1, 2 \sqrt{2}, 0)$,  $\langle v_9 |= (-\sqrt{2}, -1, \sqrt{3})$, $\langle v_{10} |= (0, \sqrt{3}, 1)$, $\langle v_{11} |= (2\sqrt{2}, -1, \sqrt{3})$, $\langle v_{12} |= (0,-1,1)$, $\langle v_{13} |= (2\sqrt{2}, -1, -3)$, $\langle v_{14} |= (\sqrt{2}, 1, 1)$.}
\end{figure}

\begin{figure}[t] 
\label{fig:graph-monotone}
\centering
\includegraphics[width=7cm]{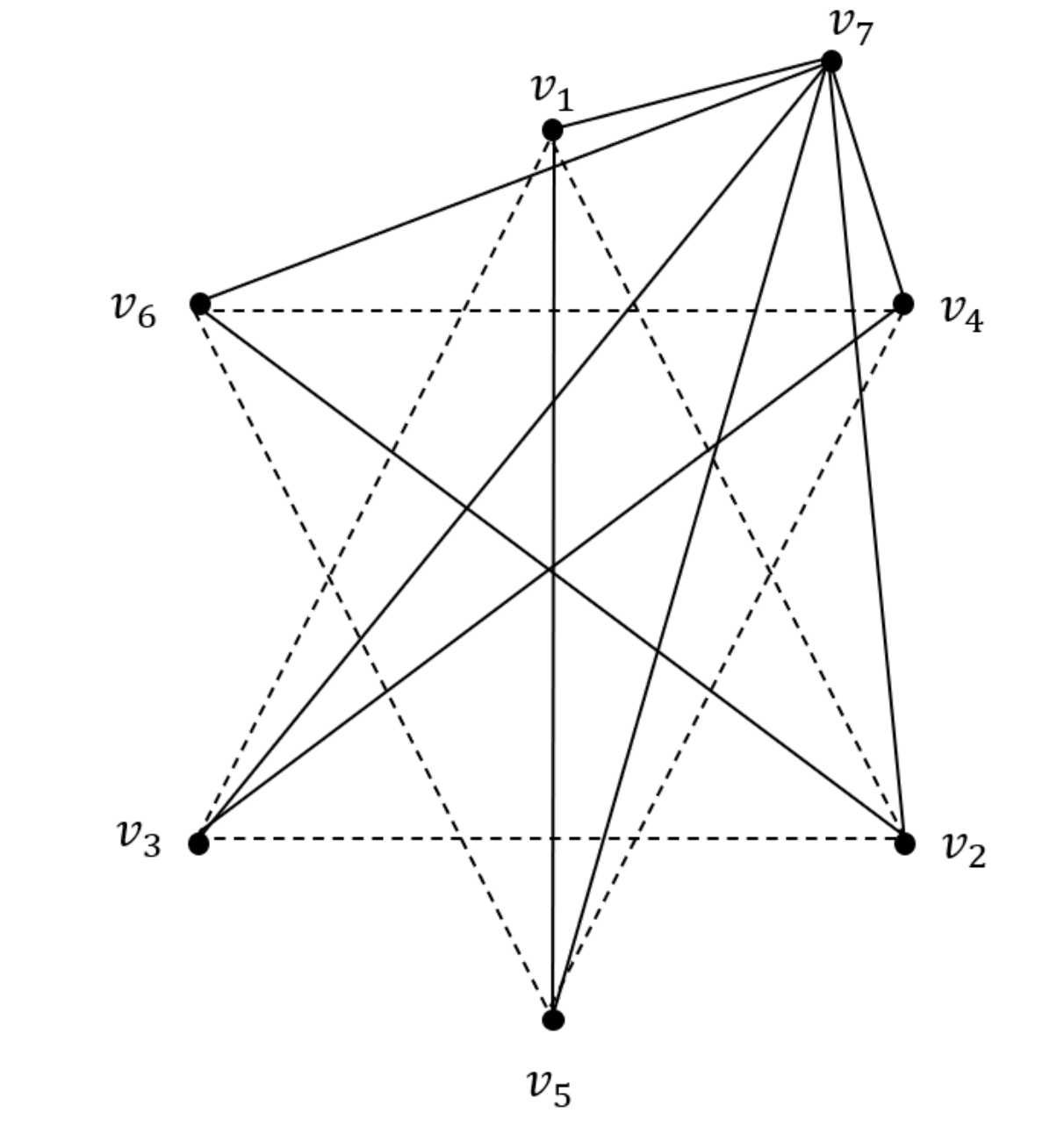}
 \caption{The proof of Proposition \ref{prop:graph-monotone} is based on this graph. Each of the dotted edges in the figure represents a copy of the $01$-gadget from Fig. 2.}
\end{figure}
\begin{prop}
\label{prop:graph-monotone}
The graph property $P^{d,n}$ of all the graphs on $n$ vertices which admit a faithful orthogonal representation in $\mathbb{R}^d$ is not monotone-decreasing.
\end{prop}
\begin{proof}
The proof comes from a specific $01$-gadget construction shown in Fig. 2. This gadget has representation in $\mathbb{R}^3$, and moreover the two distinguished vectors $|v_1 \rangle$ and $|v_2 \rangle$ have overlap at most $1/2$ in this dimension, i.e., $|\langle v_1 | v_2 \rangle| \leq 1/2$, for a proof see \cite{RRHP+20}. We use this gadget to construct the graph $G$ shown in Fig. 3. In the graph $G$ which has a representation in $\mathbb{R}^3$, the pairs of vectors $| v_1 \rangle, |v_2 \rangle$, $|v_2 \rangle, |v_3 \rangle$ and $|v_3 \rangle, |v_1 \rangle$ are connected by the $01$-gadget, as well as the pairs of vectors $|v_4 \rangle, |v_5 \rangle$, $|v_5 \rangle, |v_6 \rangle$ and $|v_6 \rangle, |v_4 \rangle$. Furthermore, the six vectors $|v_1 \rangle, \dots, |v_6 \rangle$ all lie on a single plane since they are all distinct and connected by edges to vector $|v_7 \rangle$ which we may take without loss of generality to be $|v_7 \rangle = (0,0,1)^T$. It then follows from the aforementioned property of the $01$-gadget that the faithful orthogonal representation in $\mathbb{R}^3$ is only possible if the maximum overlap $|\langle v_i | v_j \rangle| = 1/2$ is achieved for every pair of vectors $|v_i \rangle, |v_j \rangle$ connected by the $01$-gadget. 

Now, the edge $v_2, v_6$ (and also each of the edges $v_1, v_5$ and $v_3, v_4$) ensures that the overlaps $|\langle v_1 | v_4 \rangle| = \sqrt{3}/2$, $|\langle v_2 | v_5 \rangle| = \sqrt{3}/2$ and $|\langle v_3 | v_6 \rangle| = \sqrt{3}/2$ are rigid, i.e., these overlaps are necessary for the representation in $\mathbb{R}^3$. Deleting one of these edges, say $v_2, v_6$ then means that the resulting graph $G' = G \setminus (v_2,v_6)$ no longer has a faithful orthogonal representation in $\mathbb{R}^3$. In other words, for the graph obtained by the edge deletion operation we have $d_R^*\left(G \setminus (v_2,v_6) \right) > 3$. Therefore, we conclude that the graph property $P^{d,n}$ of all the graphs on $n$ vertices which admit a faithful orthogonal representation in $\mathbb{R}^d$ is not monotone-decreasing.

\end{proof}

Note that the minimum dimension is clearly also not monotone-increasing, one can readily find examples of graphs in which deleting an edge reduces the minimum dimension. Take for example the graph $G$ consisting of four vertices $u_1, u_2, u_3, u_4$ in the square configuration with edge set $E(G) = \{(u_1,u_2), (u_2,u_3), (u_3,u_4), (u_4, u_1 \}$, this graph has $d_R^*(G) = 4$, while the deletion of an edge $(u_4,u_1)$ reduces the minimum dimension, i.e., $d_R^*(G \setminus (u_4,u_1)) = 3$. We therefore infer that the minimum dimension of a faithful orthogonal representation is neither monotone-increasing nor monotone-decreasing, i.e., is not graph monotone. 

\subsection{Large violations in $01$-gadgets via a one-to-one connection with two-player Hardy paradoxes}
In this subsection, we show a method to construct $01$-gadgets by establishing a one-to-one connection with the non-locality proofs known as two-player Hardy paradoxes. In particular, we show that for every two-player Hardy paradox for a maximally entangled state in $\mathbb{C}^d \otimes \mathbb{C}^d$, the union of the sets of projectors measured by the two players constitutes a $01$-gadget contextuality proof in $\mathbb{C}^d$. The converse statement is already well known, see for instance \cite{RRHP+20}. Namely, it is possible to construct a Hardy proof of non-locality in $\mathbb{C}^d \otimes \mathbb{C}^d$ by considering the Bell scenario in which each of two parties performs the projective measurements corresponding to any $01$-gadget in $\mathbb{C}^d$.   

Let us first give a general definition of a two-player Hardy paradox suited to our purpose. The general Hardy paradox consists of two parts: (i) a set of Hardy constraints $\texttt{HC}$ that impose that the probability of a certain set of events (input-output combinations) is zero, and (ii) a particular Hardy probability which can be inferred to be zero in all classical (local hidden variable) theories from the Hardy constraints $\texttt{HC}$, while being non-zero for a particular choice of quantum entangled state and measurements that still satisfy $\texttt{HC}$.

\begin{definition}
Let $\mathcal{H} = \mathcal{H}_1 \otimes \mathcal{H}_2$ and let $| \psi \rangle \in \mathcal{H}$ be a pure state. A two-player Hardy paradox with respect to $|\psi \rangle$ is a pair $(B_1, B_2)$ where $B_i$ is a set of orthonormal bases of $\mathcal{H}_i$ such that the following holds. 

Let $\texttt{HC}$ be the set of Hardy constraints defined on $B_1 \times B_2$ as follows. $\texttt{HC}((b_1, b_2))$ is the set of pairs $(u_1, u_2) \in b_1 \times b_2$ satisfying $\langle \psi | u_1, u_2 \rangle = 0$, i.e., the measurement outcome $(u_1, u_2)$ has zero probability of occurring if $|\psi \rangle$ is measured with respect to the basis $b_1 \times b_2$ of $\mathcal{H}$. 

Then there exists a basis $b_1^* \times b_2^*$ of $\mathcal{H}$ and a pair $(u_1^*, u_2^*) \in b_1^* \times b_2^*$ obeying $\langle \psi | u_1^*, u_2^* \rangle \neq 0$ such that for every pair of functions $(s_1, s_2)$, where $s_i$ is defined on $B_i$ and $s_i(b_i) \in b_i$ for all $b_i \in B_i$ for which $\left(s_1(b_1), s_2(b_2) \right) \in \texttt{HC}\left((b_1, b_2) \right)$ for all $(b_1, b_2) \in (B_1, B_2)$, it holds that $\left(s_1(b_1^*), s_2(b_2^*) \right) \neq (u_1^*, u_2^*)$.

\end{definition}

Let us first recall the following Proposition \ref{prop:Hardy-gadget1} which states that one can construct a two-player Hardy paradox with respect to the maximally entangled state in $\mathbb{C}^d \otimes \mathbb{C}^d$ in the Bell scenario in which Alice and Bob each perform measurements given by the projectors forming a $01$-gadget in $\mathbb{C}^d$. The proof is by construction, and is given in \cite{RRHP+20}.

\begin{prop}
\label{prop:Hardy-gadget1}
Let $\mathcal{H} = \mathbb{C}^d$, $S \subseteq \mathcal{H}$, and let $B = \big\{ b \subseteq S | \; b \; \text{is an orthonormal basis of} \; \mathcal{H} \big\}$. Consider the state $|\psi \rangle = \frac{1}{\sqrt{d}} \left(|0,0\rangle + |1,1 \rangle + \dots + |d-1, d-1 \rangle \right) \in \mathcal{H} \otimes \mathcal{H}$. If $S$ is a $01$-gadget in $\mathcal{H}$, then $(B, \bar{B})$ is a Hardy paradox with respect to $| \psi \rangle$. 
\end{prop}

We now move to the result of this subsection. Namely, for any Hardy paradox with respect to the maximally entangled state in $\mathbb{C}^d \otimes \mathbb{C}^d$, the union of the sets of projectors measured by Alice and Bob constitutes a $01$-gadget in $\mathbb{C}^d$. 

\begin{prop}
\label{prop:Hardy-gadget2}
Let $\mathcal{H} = \mathbb{C}^d$, let $B_1$ and $B_2$ be two orthonormal bases of $\mathcal{H}$, and let $|\psi \rangle = \frac{1}{\sqrt{d}} \left(|0,0\rangle + |1,1 \rangle + \dots + |d-1, d-1 \rangle \right) \in \mathcal{H} \otimes \mathcal{H}$. Let $S$ be the set $S := \bigcup_{b \in B_1} b \cup \bigcup_{b \in B_2} \bar{b} \subseteq \mathcal{H}$. If $(B_1, B_2)$ is a Hardy paradox with respect to $| \psi \rangle$, then $S$ is a $01$-gadget in $\mathcal{H}$.
\end{prop}
\begin{proof}
The proof follows a similar idea linking KS proofs to pseudo-telepathy games in \cite{RW04}. Let $f : S \rightarrow \{0,1\}$ be a function such that for all orthonormal bases $b \subseteq S$, we have $\sum_{v \in b} f(v) = 1$. 

Since $(B_1, B_2)$ is a Hardy paradox with respect to the maximally entangled state $| \psi \rangle$, there exist bases $b_1^*\times b_2^*$ of $\mathbb{C}^d \otimes \mathbb{C}^d$ and a pair $(u_1^*, u_2^*) \in b_1^* \times b_2^*$ obeying $|\langle \psi  | u_1^*, u_2^* \rangle| \neq 0$ such that for every pair of functions $(s_1, s_2)$, where $s_i$ is defined on $B_i$ and $s_i(b_i) \in b_i$ for all $b_i \in B_i$ for which $\left(s_1(b_1), s_2(b_2) \right) \in \texttt{HC}\left((b_1, b_2) \right)$ for all $(b_1, b_2) \in (B_1, B_2)$, it holds that $\left(s_1(b_1^*), s_2(b_2^*) \right) \neq (u_1^*, u_2^*)$. 

Now, $\left(s_1(b_1^*), s_2(b_2^*) \right) \neq (u_1^*, u_2^*)$ implies that this output pair occurs with probability $0$ in all classical (local hidden variable) theories. In other words $f\left(u_1^* \right) + f\left(u_2^* \right) \leq 1$. On the other hand, we have
\begin{eqnarray}
|\langle \psi  | u_1^*, u_2^* \rangle| &\neq & 0, \nonumber \\
\frac{1}{\sqrt{d}} \sum_{i =0}^{d-1} \langle i | u_1^* \rangle \langle i | u_2^* \rangle &\neq & 0 \; \; \implies
\frac{1}{\sqrt{d}} \overline{\langle u_1^* | \overline{u_2^*} \rangle}  \neq  0. \nonumber \\
\end{eqnarray}
In other words, the two vectors $|u_1^* \rangle$ and $|u_2^* \rangle$ are non-orthogonal vectors such that in any classical deterministic assignment $f$ of the vectors in the set $S$, it holds that $f\left(u_1^* \right) + f\left(u_2^* \right) \leq 1$. Therefore, the set of vectors in $S$ constitutes a $01$-gadget with $|u_1^* \rangle$ and $|u_2^* \rangle$ being the two distinguished non-orthogonal vectors that cannot both be assigned value $1$ in any non-contextual assignment. 
\end{proof}
The results in this subsection, linking two fundamental aspects of non-locality and contextuality, are useful in the construction of novel $01$-gadgets and Hardy paradoxes with properties suited to specific applications. In the Appendix II, we use this connection to give a novel construction of a $01$-gadget contextuality proof for any arbitrary two vectors in $\mathbb{C}^d$, in particular with the overlap between the two distinguished vectors being arbitrarily close to unity. 

Similarly, we have by Proposition \ref{prop:Hardy-gadget1} that when the two parties Alice and Bob perform, on the maximally entangled state in $\mathbb{C}^d \otimes \mathbb{C}^d$, projective measurements corresponding to a $01$-gadget with the two distinguished vectors being identical, then this gives rise to a Hardy paradox with the Hardy probability taking the value $1/d$. One possible method to achieve this is through the $01$-gadget based on the following vectors \cite{RRHP+20}:
\begin{eqnarray}
&&|u_1 \rangle = (1,-1,0)^T; \; \; |u_2 \rangle = (1,1,1)^T;\nonumber \\
&&|u_3 \rangle = (1,1,0)^T; \; \; |u_4 \rangle = (1,1,b)^T; \nonumber \\
&&|u_5 \rangle = (-2,1,1)^T; \; \; |u_6 \rangle = (1,-1,3)^T; \nonumber \\
&&|u_7 \rangle = (3,-3,-2)^T; \; \; |u_8 \rangle = (2,0,3)^T;\nonumber \\
&& |u_9 \rangle = (-3,0,2)^T; \;\;|u_{10} \rangle = (-2,2,-3)^T; \nonumber \\
&& |u_{11} \rangle = (3,-3,-4)^T; \; \; |u_{12} \rangle = (4,0,3)^T; \nonumber \\
&&|u_{13} \rangle = (-3,0,4)^T; \; \; |u_{14} \rangle = (-4,4,-3)^T;  \nonumber \\
&&|u_{15} \rangle = (3,-3,-8)^T; \; \; |u_{16} \rangle = (8,0,3)^T; \nonumber \\
&& |u_{17} \rangle = (-3,0,8)^T; \;\;|u_{18} \rangle = (-8,4+\sqrt{7},-3)^T;\nonumber \\
&&|u_{19} \rangle = (0,1,-1)^T; |u_{20} \rangle = (0,1,0)^T; \nonumber \\
&& |u_{21} \rangle = (0,-3+8b,-16-3b)^T; \;\;|u_{22} \rangle = (1,0,0)^T; \nonumber \\
&&|u_{23} \rangle = (1,0,-1)^T; \;\; |u_{24} \rangle = (2-\sqrt{2},0,1)^T; \nonumber \\
&&|u_{25} \rangle = (1,-2,1)^T; \; \; |u_{26} \rangle = (0,1,2)^T; \nonumber \\
&& |u_{27} \rangle = (0,2,-1)^T; \;\; |u_{28} \rangle = (1,-1,-2)^T; \nonumber \\
&& |u_{29} \rangle = (1,-1,1)^T; \; \; |u_{30} \rangle = (0,1,1)^T; \nonumber \\
&&|u_{31} \rangle = (0,1,-1)^T; \;\; |u_{32} \rangle = (-1,1,1)^T; \nonumber \\
&&|u_{33} \rangle = (-1,1,-2)^T; \;\; |u_{34} \rangle = (0,2,1)^T; \nonumber \\
&& |u_{35} \rangle = (0,1,-2)^T; \; \; |u_{36} \rangle = (2,-2,-1)^T; \nonumber \\
&&|u_{37} \rangle = (1,-1,4)^T; \; \;  |u_{38} \rangle = (-2-\sqrt{2},6-\sqrt{2},2)^T;  \nonumber \\
&& |u_{39} \rangle = |u_2 \rangle;\; \; |u_{40} \rangle = |u_3 \rangle; \; \; |u_{41} \rangle = (1,1,-2+\sqrt{2})^T; \nonumber \\
&& |u_{42} \rangle = |u_1 \rangle; |u_{43} \rangle = (0,0,1)^T; \nonumber 
\end{eqnarray}
with $b = \frac{-4+\sqrt{7}}{3}$, and where we have the following identities $|u_{1}\rangle=|u_{42}\rangle$, $|u_2\rangle=|u_{39}\rangle$,  $|u_3\rangle=|u_{40}\rangle$. The vector $|u_1 \rangle = |u_{42} \rangle$ forms the distinguished vector of this gadget, it is an open question whether this set of $40$ vectors is the minimal set with this property. 

\subsection{Application I. A rigidity property of $01$-gadgets and randomness certification.}
As we have noted (also see \cite{RRHP+20}), the state-dependent tests of contextuality known as $01$-gadgets are especially useful for certification of private intrinsic quantum randomness, under the assumption of a fixed Hilbert space dimension of the system, and the assumption that the same projectors are being measured under different contexts. We leave the formulation of a contextuality-based randomness generation protocol and its rigorous security proof for future work \cite{R+21}. Here, we show a rigidity property of the $01$-gadget which is especially useful in contextuality-based randomness certification. In particular, we show that for the Clifton bug, the observation of the probabilities $P(|v_1 \rangle) = 1$ and $P(|v_8 \rangle) = 1/9$ implies a rigidity statement on the projectors being measured. Up to the left-right symmetry of the Clifton bug, and up to unitaries, the set of projectors realizing the Clifton bug in $\mathbb{R}^3$ is unique.

\begin{prop}
In the measurement configuration of the Clifton $01$-gadget in Fig. 1, the observation of the probabilities $P(|v_1 \rangle) = 1$ and $P(|v_8 \rangle) = 1/9$ implies that the set of projectors realizing the Clifton bug in $\mathbb{R}^3$ is unique (up to unitaries and the relabelling $\{u_2 \leftrightarrow u_3, u_4 \leftrightarrow u_5, u_6 \leftrightarrow u_7 \}$). 
\end{prop}
\begin{proof}
The conditions $P(|v_1 \rangle) = 1$ and $P(|v_8 \rangle) = 1/9$ guarantee that without loss of generality, we may take
\begin{eqnarray}
|v_1 \rangle = (1,0,0)^T, \quad |v_8 \rangle = \left(\frac{1}{3}, \frac{2\sqrt{2}}{3}, 0 \right)^T,
\end{eqnarray}
and the state being measured to be $|\psi \rangle = |v_1 \rangle = (1,0,0)^T$. Parametrizing 
\begin{eqnarray}
|v_2 \rangle &=& \left(0, \cos(\theta_1), \sin(\theta_1) \right)^T, \nonumber \\
|v_3 \rangle &=& \left(0, \cos(\theta_2), \sin(\theta_2) \right)^T,
\end{eqnarray}
to maintain orthogonality with $|v_1 \rangle$, we may deduce
\begin{eqnarray}
|\tilde{v}_6 \rangle &=& |v_2 \rangle \times |v_8 \rangle = \left( \frac{2 \sqrt{2} \sin(\theta_1)}{3}, \frac{-\sin(\theta_1)}{3}, \frac{\cos(\theta_1)}{3} \right)^T, \nonumber \\
|\tilde{v}_7 \rangle &=& |v_3 \rangle \times |v_8 \rangle = \left( \frac{2 \sqrt{2} \sin(\theta_2)}{3}, \frac{-\sin(\theta_2)}{3}, \frac{\cos(\theta_2)}{3} \right)^T, \nonumber \\
\end{eqnarray}
where the vectors $|\tilde{v}_6 \rangle$ and $|\tilde{v}_7 \rangle$ are unnormalized, i.e., $|v_6 \rangle = \frac{|\tilde{v}_6 \rangle}{\| |\tilde{v}_6 \rangle \|}$ and $|v_7 \rangle = \frac{|\tilde{v}_7 \rangle}{\| |\tilde{v}_7 \rangle \|}$. Finally, we obtain 
\begin{eqnarray}
|\tilde{v}_4 \rangle &=& |v_2 \rangle \times |v_6 \rangle = \left(-1, -2\sqrt{2} \sin^2(\theta_1), \sqrt{2} \sin(2 \theta_1) \right)^T, \nonumber \\
|\tilde{v}_5 \rangle &=& |v_3 \rangle \times |v_7 \rangle = \left(-1, -2\sqrt{2} \sin^2(\theta_2), \sqrt{2} \sin(2 \theta_2) \right)^T. \nonumber \\
\end{eqnarray}
Imposing the remaining orthogonality condition $|\langle v_4 | v_5 \rangle | = 0$, we obtain that
\begin{eqnarray}
1 + 8 \sin^2(\theta_1) \sin^2(\theta_2) + 2 \sin(2 \theta_1) \sin(2 \theta_2) = 0, 
\end{eqnarray}
or equivalently that
\begin{eqnarray}
3 - 2 \cos(2 \theta_1) - 2 \cos(2 \theta_2) + 2 \cos\left(2(\theta_1 - \theta_2) \right) = 0.
\end{eqnarray}
We now want to show that the solution set $\{\theta_1, \theta_2 \}$ to this equation is unique (up to the symmetry of the relabelling $\{u_2 \leftrightarrow u_3, u_4 \leftrightarrow u_5, u_6 \leftrightarrow u_7 \}$). We write $\theta_1 = \beta - \alpha$ and $\theta_2 = \beta + \alpha$ for $\alpha, \beta \in [0, \pi]$ to obtain
\begin{eqnarray}
\label{eq:alphbet}
1 + 4 \cos^2(2 \alpha) - 4 \cos(2 \alpha) \cos(2 \beta) = 0.
\end{eqnarray}
Using the facts that $\left(1 \pm 2 \cos(2 \alpha) \right)^2 \geq 0$, we deduce that $\frac{1 + 4 \cos^2(2 \alpha)}{4 \cos(2 \alpha)} \geq 1$ or $\frac{1 + 4 \cos^2(2 \alpha)}{4 \cos(2 \alpha)} \leq -1$. This gives that the Eq.(\ref{eq:alphbet}) can only be satisfied when $\beta = 0$ (with corresponding $\alpha = \pi/6$) or $\beta = \pi/2$ (with corresponding $\alpha = \pi/3$). The solutions are thus $(\theta_1, \theta_2) = (5\pi/6, \pi/6)$ and $(\theta_1, \theta_2) = (\pi/6, 5\pi/6)$. In other words, up to the relabelling (and unitary rotations of all the vectors), the set of vectors realizing the Clifton bug in $\mathbb{R}^3$ is unique.  
\end{proof}

\section{Conclusion and Open Questions}\label{sec:concl}
In this paper, we have studied large violations of state-independent and special state-dependent non-contextuality inequalities. In particular, we exploited a connection between Kochen Specker proofs and pseudo-telepathy games from \cite{RW04} to show KS sets with large violations in Hilbert spaces of dimension $d \geq 2^{17}$. We also showed that these KS sets satisfy the condition for a boost in one-shot zero-error capacity through shared entanglement, providing an interesting class of channels for this application. Intriguingly, we showed that despite their foundational interest, the maximum violation of KS state-independent non-contextuality inequalities does not serve to certify intrinsic randomness. To remedy this fault, we studied large violations in a class of state-dependent non-contextuality inequalities known as $01$-gadgets.We showed an intriguing technical result here that the minimum dimension of a faithful orthogonal representation of a graph in $\mathbb{R}^d$ is not graph monotone, a result that has fundamental applications in constructions of KS proofs. We derived a one-to-one correspondence between $01$-gadgets in $\mathbb{C}^d$ and two-player Hardy paradoxes for the maximally entangled state in $\mathbb{C}^d \otimes \mathbb{C}^d$ and exploited this connection to construct novel large violation $01$-gadgets. Finally, we show that $01$-gadgets are natural candidates for self-testing under the assumption of a fixed dimension, a result which has interesting application for contextuality-based randomness generation which we study in future work. 

In future, it would be interesting to study the self-testing properties of $01$-gadgets \cite{BRV+19} and derive rigorous security proofs of contextuality-based randomness certification protocols using these. Similarly, the connection to two-player Hardy paradoxes serves as a natural tool to construct optimal randomness amplification protocols \cite{RBHH+15, BRGH+16, WBGH+16}. It is also interesting to investigate applications of large violations of the non-contextuality inequalities studied here for one-way communication problems \cite{SHP19}. Finally, it would be good to phrase the large violation results in a resource-theoretic framework, in particular to show that the non-contextual simulation of the large violation demands a correspondingly large classical memory \cite{KGPL+11}.

\textit{Acknowledgments.-}
We acknowledge useful discussions with Susan Liao Jiayang. This work is supported by the Start-up Fund 'Device-Independent Quantum Communication Networks' from The University of Hong Kong, the Seed Fund 'Security of Relativistic Quantum Cryptography' and the Early Career Scheme (ECS) grant 'Device-Independent Random Number Generation and Quantum Key Distribution with Weak Random Seeds'.

\begin{widetext}
\section{Appendix I: The number of neighbors of each vertex in the large violation KS graph}
In this Appendix, we show a structural property of the large violation KS graph, that is useful in deriving large separations in the success probability in communication tasks such as shown in \cite{SHP19}. We divide the vertices in the KS graph into two sets arising from the non-local hidden matching game, the vectors from Alice's optimal quantum strategy form the bases $P$ and the vectors from Bob's optimal strategy form the bases $Q$. The KS set $S$ is then the union of the two bases sets, i.e., $S = P \cup Q$.
The number of bases in $Q$ is $\frac{n}{2}$ and the number of vectors in bases set $Q$ is $\frac{n^2}{2}$. The number of neighbours for an arbitrary vector $Q_M^{(i,j)}$ can be derived as follows.

\textbf{Case 1 in basis $Q_M$}: All $n-1$ other vectors are orthogonal to $Q_M^{(i,j)}$.

\textbf{Case 2 in basis $Q_{M'}$}: All vectors except the special four vectors in basis $Q_{M'}$ are orthogonal to $Q_M^{(i,j)}$. The special four vectors are 
\begin{center}
$Q_{M'}^{(i,j')}=\bordermatrix{
        &     \cr
        & 0    \cr
        & \vdots    \cr
       i & 1   \cr
       & 0    \cr
      & \vdots    \cr
      j'& 1    \cr
      &0   \cr
     &  \vdots  \cr
}$
\qquad 
$Q_{M'}^{(i,j')}=\bordermatrix{
        &     \cr
        & 0    \cr
        & \vdots    \cr
       i & -1   \cr
       & 0    \cr
      & \vdots    \cr
      j'& 1    \cr
      &0   \cr
     &  \vdots  \cr
}$
\qquad 
$Q_{M'}^{(i',j)}=\bordermatrix{
        &     \cr
        & 0    \cr
        & \vdots    \cr
       i' & 1   \cr
       & 0    \cr
      & \vdots    \cr
      j& 1    \cr
      &0   \cr
     &  \vdots  \cr
}$
\qquad 
$Q_{M'}^{(i',j)}=\bordermatrix{
        &     \cr
        & 0    \cr
        & \vdots    \cr
       i' & 1   \cr
       & 0    \cr
      & \vdots    \cr
      j& -1    \cr
      &0   \cr
     &  \vdots  \cr
}$
\end{center}
and the $i',j'$ of these vectors depend on the matching $M'$. In total, there are $(\frac{n}{2}-1)(n-4)$ vectors orthogonal to $Q_M^{(i,j)}$.

\textbf{Case 3 in bases set $P$}: Consider $Q_M^{(i,j)}$ as above with entries $1$ in positions $i$ and $j$ and entry $0$ otherwise. 
Note that vectors in bases set $P$ are the vectors with every entry to be arranged 1 or -1. So for vectors
$p=\left(\begin{array}{c}
[1] \cr
[2] \cr
\vdots \cr
[i] \cr
\vdots \cr
[j] \cr
\vdots \cr
[n] \cr
\end{array}\right),
$
if $Q_M^{(i,j)}$ is orthogonal to $p$, we have $[i]+[j]=0$. There are a total of $\frac{2^{n-2}\times 2}{2}$ vectors in bases set $P$ that meet this requirement.


\textbf{Consider these three cases, the number of neighbour of $Q_M^{(i,j)}$ is} $(n-1)+(\frac{n}{2}-1)(n-4)+\frac{2^{n-2}\times 2}{2}=2^{(n-2)}+\frac{n^{2}}{2}-2n+3$.


\textbf{Case 1 in bases set $P$}: For $p=\left(\begin{array}{c}
[1] \cr
[2] \cr
\vdots \cr
[i] \cr
\vdots \cr
[j] \cr
\vdots \cr
[n] \cr
\end{array}\right)
$ with each entry equal to 1 or -1, there must be a vector $p'=\left(\begin{array}{c}
[1]' \cr
[2]' \cr
\vdots \cr
[i]' \cr
\vdots \cr
[j] '\cr
\vdots \cr
[n]' \cr
\end{array}\right)$ in bases set $P$ which is orthogonal to $p$.
So that $$[1][1]'+[2][2]'+\cdots+[i][i]'+\cdots+[n][n]'=0$$
For half of the terms in this equation $[i][i]'=1$, and for the other half $[i][i]'=-1$.

In order to find all the neighbours of $p$, we first consider inverting two entries of vector $p'$ at the same time. And these two entries must come from different terms sets, i.e., one entry $[i]'$ of $p'$ is from $[i][i]'=1$ and correspondingly another entry $[i]'$ of $p'$ is from $[i][i]'=-1$. We then flip four entries of vector $p'$ simultaneously, two from $[i][i]'=1$, another two from $[i][i]'=-1$. We then proceed to flip six, eight entries and so on, until we invert all the entries of $p'$.
In total, the number of neighbours of $p$ is thus found to be
$$
\frac{1+C_{\frac{n}{2}}^1C_{\frac{n}{2}}^1+C_{\frac{n}{2}}^2C_{\frac{n}{2}}^2+C_{\frac{n}{2}}^3C_{\frac{n}{2}}^3+\cdots+C_{\frac{n}{2}}^{\frac{n}{2}}C_{\frac{n}{2}}^{\frac{n}{2}}}{2}=\frac{C_n^{\frac{n}{2}}}{2}
$$


\textbf{Case 2 in bases set $Q$}:
In basis $Q_M$, since either $[i]+[j]=0$ or $[i]-[j]=0$, either $Q_M^{(i,j)}$ or $Q_M^{'(i,j)}$ will be orthogonal to $P_x^i$. In basis $Q_M$, half of the vectors are orthogonal to $p$. In bases set $Q$, $(\frac{n}{2})(\frac{n}{2})=\frac{n^2}{4}$ vectors are orthogonal to $p$. In total, the number of neighbours of $p$ is thus found to be $\frac{C_n^{\frac{n}{2}}}{2}+\frac{n^2}{4}$.

\section{Appendix II: Large violation $01$-gadgets from Hardy paradoxes}

In this Appendix, we present a construction of a $01$-gadget in which the two distinguished vectors, denoted here by $A_{4k}^0$ and $B_{4k}^0$ exhibit arbitrary separation as $k \rightarrow \infty$, i.e., with $0 \leq |\langle A_{4k}^) | B_{4k}^0  \rangle| \leq 1$ as $k \rightarrow \infty$. The construction makes use of a version of the ladder proof of Hardy non-locality from \cite{Cab5} and a one-to-one correspondence between Hardy non-locality and $01$-gadget contextuality shown in the main text. The ladder proof in \cite{Cab5} shows a version of Hardy non-locality for the maximally entangled state of two spin-$1$ particles. By the Proposition \ref{prop:Hardy-gadget1} and \ref{prop:Hardy-gadget2} shown in the main text, the projectors measured by the two parties in the Hardy paradox can be converted into a $01$-gadget proof of contextuality for a single spin-$1$ particle.

The $01$-gadget construction is shown in Fig. \ref{fig:gadg-ladder} and the optimal vectors are given as:

$$ A_{4k}^0=\left(\begin{array}{c}
1 \\
0 \\
-1
\end{array}\right)\qquad \qquad 
B_{4k}^0=\left(\begin{array}{c}
-\cos 2 \phi+i \sin 2 \phi \\
0 \\
1
\end{array}\right)
$$

$$
A_{4k-1}^0=\left(\begin{array}{c}
e^{-i2\phi} \\
\sqrt{2} \cot \theta_1 \cos^2\phi(\tan\phi+i) \\
1
\end{array}\right)\qquad \qquad 
B_{4k-1}^0=\left(\begin{array}{c}
1\\
-i \sqrt{2} \tan\theta_1 \\
1
\end{array}\right)$$

\begin{figure*}[t] 
\label{fig:gadg-ladder}
\centering
\includegraphics[width=18cm]{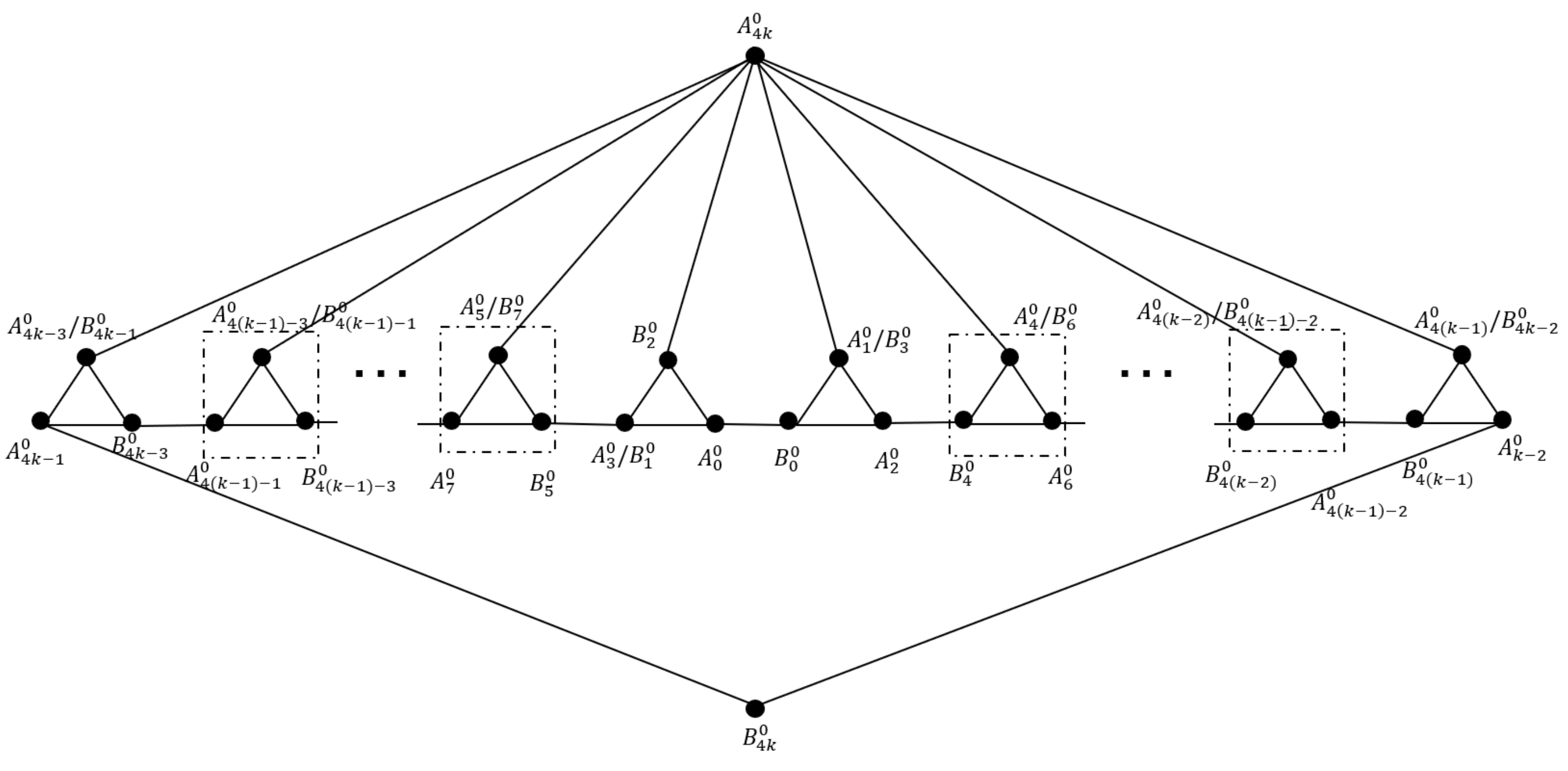}
 \caption{\label{4} The construction of a $01$-gadget where the angle between the two distinguished vectors $A_{4k}^0$ and $B_{4k}^0$ can be made arbitrary as $k \rightarrow \infty$. The construction makes use of a version of the ladder proof of Hardy non-locality from \cite{Cab5} and a one-to-one correspondence between Hardy non-locality and $01$-gadget contextuality shown in the main text.}
\end{figure*}

$$
A_{4k-2}^0=\left(\begin{array}{c}
e^{-i2\phi} \\
-\sqrt{2} \cot \theta_1 \cos^2\phi(\tan\phi+i) \\
1
\end{array}\right)\qquad \qquad 
B_{4k-2}^0=\left(\begin{array}{c}
1\\
i \sqrt{2} \tan\theta_1 \\
1
\end{array}\right)$$

$$
A_{4k-3}^0=B_{4k-1}^0\qquad \qquad 
B_{4k-3}^0=\left(\begin{array}{c}
c_{1}^{2} \sin ^{2} \theta_{1}-\cot ^{2} \phi+i2 c_{1} \sin \theta_{1} \cot \phi \\
-\sqrt{2} c_{1} \cot \phi \cos \theta_{1}+i \sqrt{2} c_{1}^{2} \sin \theta_{1} \cos \theta_{1} \\
\cot ^{2} \phi +c_{1}^{2} \sin ^{2} \theta_{1}
\end{array}\right)$$

$$
A_{4(k-1)}^0=B_{4k-2}^0\qquad \qquad 
B_{4(k-1)}^0=\left(\begin{array}{c}
c_{1}^{2} \sin ^{2} \theta_{1}-\cot ^{2} \phi+i2 c_{1} \sin \theta_{1} \cot \phi \\
\sqrt{2} c_{1} \cot \phi \cos \theta_{1}-i \sqrt{2} c_{1}^{2} \sin \theta_{1} \cos \theta_{1} \\
\cot ^{2} \phi +c_{1}^{2} \sin ^{2} \theta_{1}
\end{array}\right)$$

$$
A_{4k-5}^0=\left(\begin{array}{c}
c_{2}^{-2} \sin ^{2} \theta_{2}-\tan ^{2} \phi-i 2 c_{2}^{-1} \sin \theta_{2}\tan\phi \\
\sqrt{2} c_{2}^{-1} \tan\phi\cos \theta _{2}+i \sqrt{2} c_{2}^{-2} \sin \theta_{2} \cos \theta_{2}  \\
\tan ^{2} \phi+c_{2}^{-2} \sin ^{2} \theta_{2}
\end{array}\right) \qquad \qquad 
B_{4k-5}^0= \left(\begin{array}{c}
1\\
-i \sqrt{2} \tan\theta_2 \\
1
\end{array}\right)$$

$$
A_{4k-6}^0=\left(\begin{array}{c}
c_{2}^{-2} \sin ^{2} \theta_{2}-\tan ^{2} \phi-i 2 c_{2}^{-1} \sin \theta_{2}\tan\phi \\
-\sqrt{2} c_{2}^{-1} \tan\phi\cos \theta _{2}-i \sqrt{2} c_{2}^{-2} \sin \theta_{2} \cos \theta_{2}  \\
\tan ^{2} \phi+c_{2}^{-2} \sin ^{2} \theta_{2}
\end{array}\right) \qquad \qquad 
B_{4k-6}^0= \left(\begin{array}{c}
1\\
i \sqrt{2} \tan\theta_2 \\
1
\end{array}\right)$$

$$
A_{4k-7}^0=B_{4k-5}^0\qquad \qquad 
B_{4k-7}^0=\left(\begin{array}{c}
c_{2}^{2} \sin ^{2} \theta_{2}-\cot ^{2} \phi+i2 c_{2} \sin \theta_{2} \cot \phi \\
-\sqrt{2} c_{2} \cot \phi \cos \theta_{2}+i \sqrt{2} c_{2}^{2} \sin \theta_{2} \cos \theta_{2} \\
\cot ^{2} \phi +c_{2}^{2} \sin ^{2} \theta_{2}
\end{array}\right)$$

$$
A_{4(k-2)}^0=B_{4k-6}^0\qquad \qquad 
B_{4(k-2)}^0=\left(\begin{array}{c}
c_{2}^{2} \sin ^{2} \theta_{2}-\cot ^{2} \phi+i2 c_{2} \sin \theta_{2} \cot \phi \\
\sqrt{2} c_{2} \cot \phi \cos \theta_{2}-i \sqrt{2} c_{2}^{2} \sin \theta_{2} \cos \theta_{2} \\
\cot ^{2} \phi +c_{2}^{2} \sin ^{2} \theta_{2}
\end{array}\right)$$

$$
\vdots \qquad \qquad \qquad \qquad \qquad \qquad \qquad \qquad \vdots $$

$$
A_7^0=\left(\begin{array}{c}
c_{k-1}^{-2} \sin ^{2} \theta_{k-1}-\tan ^{2} \phi-i 2 c_{k-1}^{-1} \sin \theta_{k-1}\tan\phi \\
\sqrt{2} c_{k-1}^{-1} \tan\phi\cos \theta _{k-1}+i \sqrt{2} c_{k-1}^{-2} \sin \theta_{k-1} \cos \theta_{k-1}  \\
\tan ^{2} \phi+c_{k-1}^{-2} \sin ^{2} \theta_{k-1}
\end{array}\right) \qquad \qquad 
B_7^0= \left(\begin{array}{c}
1\\
-i \sqrt{2} \tan\theta_{k-1} \\
1
\end{array}\right)$$

$$
A_6^0=\left(\begin{array}{c}
c_{k-1}^{-2} \sin ^{2} \theta_{k-1}-\tan ^{2} \phi-i 2 c_{k-1}^{-1} \sin \theta_{k-1}\tan\phi \\
-\sqrt{2} c_{k-1}^{-1} \tan\phi\cos \theta _{k-1}-i \sqrt{2} c_{k-1}^{-2} \sin \theta_{k-1} \cos \theta_{k-1}  \\
\tan ^{2} \phi+c_{k-1}^{-2} \sin ^{2} \theta_{k-1}
\end{array}\right) \qquad \qquad 
B_6^0= \left(\begin{array}{c}
1\\
i \sqrt{2} \tan\theta_{k-1} \\
1
\end{array}\right)$$

$$
A_5^0=B_7^0\qquad \qquad 
B_5^0=\left(\begin{array}{c}
c_{k-1}^{2} \sin ^{2} \theta_{k-1}-\cot ^{2} \phi+i2 c_{k-1} \sin \theta_{k-1} \cot \phi \\
-\sqrt{2} c_{k-1} \cot \phi \cos \theta_{k-1}+i \sqrt{2} c_{k-1}^{2} \sin \theta_{k-1} \cos \theta_{k-1} \\
\cot ^{2} \phi +c_{k-1}^{2} \sin ^{2} \theta_{k-1}
\end{array}\right)$$

$$
A_4^0=B_6^0\qquad \qquad 
B_4^0=\left(\begin{array}{c}
c_{k-1}^{2} \sin ^{2} \theta_{k-1}-\cot ^{2} \phi+i2 c_{k-1} \sin \theta_{k-1} \cot \phi \\
\sqrt{2} c_{k-1} \cot \phi \cos \theta_{k-1}-i \sqrt{2} c_{k-1}^{2} \sin \theta_{k-1} \cos \theta_{k-1} \\
\cot ^{2} \phi +c_{k-1}^{2} \sin ^{2} \theta_{k-1}
\end{array}\right)$$

$$
A_3^0=\left(\begin{array}{c}
c_{k}^{-2} \sin ^{2} \theta_{k}-\tan ^{2} \phi-i 2 c_{k}^{-1} \sin \theta_{k}\tan\phi \\
\sqrt{2} c_{k}^{-1} \tan\phi\cos \theta _{k}+i \sqrt{2} c_{k}^{-2} \sin \theta_{k} \cos \theta_{k}  \\
\tan ^{2} \phi+c_{k}^{-2} \sin ^{2} \theta_{k}
\end{array}\right) \qquad \qquad 
B_3^0= \left(\begin{array}{c}
1\\
i \sqrt{2} \tan\theta_k \\
1
\end{array}\right)$$

$$
A_2^0=\left(\begin{array}{c}
c_{k}^{-2} \sin ^{2} \theta_{k}-\tan ^{2} \phi-i 2 c_{k}^{-1} \sin \theta_{k}\tan\phi \\
-\sqrt{2} c_{k}^{-1} \tan\phi\cos \theta _{k}-i \sqrt{2} c_{k}^{-2} \sin \theta_{k} \cos \theta_{k}  \\
\tan ^{2} \phi+c_{k}^{-2} \sin ^{2} \theta_{k}
\end{array}\right) \qquad \qquad 
B_2^0= \left(\begin{array}{c}
1\\
-i \sqrt{2} \tan\theta_k \\
1
\end{array}\right)$$

$$
A_1^0=B_3^0\qquad \qquad 
B_1^0=A_3^0
$$

$$
A_0^0=\left(\begin{array}{c}
c_{k}^{2} \sin ^{2} \theta_{k}-\cot ^{2} \phi+i2 c_{k} \sin \theta_{k} \cot \phi \\
-\sqrt{2} c_{k} \cot \phi \cos \theta_{k}+i \sqrt{2} c_{k}^{2} \sin \theta_{k} \cos \theta_{k} \\
\cot ^{2} \phi +c_{k}^{2} \sin ^{2} \theta_{k}
\end{array}\right) \qquad \qquad 
B_0^0=\left(\begin{array}{c}
c_{k}^{2} \sin ^{2} \theta_{k}-\cot ^{2} \phi+i2 c_{k} \sin \theta_{k} \cot \phi \\
\sqrt{2} c_{k} \cot \phi \cos \theta_{k}-i \sqrt{2} c_{k}^{2} \sin \theta_{k} \cos \theta_{k} \\
\cot ^{2} \phi +c_{k}^{2} \sin ^{2} \theta_{k}
\end{array}\right)
$$


In the ladder proof, $|\psi\rangle$ is the maximally entangled state. And $P_{\psi}\left(A_{4k}, B_{4k}\right)=\frac{1}{3} \cos^{2} \phi$. $P_{\psi}\left(A_{0}, B_{0}\right)=0$, if $\cot^{2}\phi=c_{k}^{2}\left(\cos ^{2} \theta_{k}-\sin ^{2} \theta_{k}\right)$. The coefficients $c_j$ are: $c_{1}=\sin \theta_{1}$, $c_{j+1}=c_{j} \cos \left(\theta_{j+1}-\theta_{j}\right)$ or $c_{j}=\sin \theta_{1} \prod_{k=1}^{j-1} \cos \left(\theta_{k+1}-\theta_{k}\right)$.
\end{widetext}

\end{document}